\documentclass[12pt]{article}

\usepackage[utf8]{inputenc}

\usepackage[longnamesfirst]{natbib}
\usepackage{mathtools, amsfonts, amssymb,amsthm,dsfont, color, xcolor, tikz}
\usetikzlibrary{shapes.geometric, arrows, positioning, fit, calc, shapes.misc} 
\mathtoolsset{showonlyrefs}

\usepackage{comment, subcaption, graphicx, wrapfig}
\usepackage{eurosym} 
\usepackage{ifthen}
\usepackage{hyperref}
\usepackage[ruled,vlined]{algorithm2e}
\usepackage{enumerate}

\usepackage{bbm}

\usepackage[inner=30mm,outer=20mm,top=30mm,bottom=20mm]{geometry} 
\newcounter{prop}
\newtheorem{proposition}[prop]{Proposition}
\newcounter{cor}
\newtheorem{corollary}[cor]{Corollary}
\newcounter{lem}
\newtheorem{lemma}[lem]{Lemma}

\newcounter{rem}
\newtheorem{remark}[rem]{Remark}

\newtheorem{assumption}{Assumption}

\newtheorem{algo}{Algorithm}

\newcommand{\EE}{\mathbb{E}}

\newcommand{\PP}{\mathbb{P}}

\newcommand{\x}{{\boldsymbol x}}

\newcommand{\z}{\boldsymbol z}
\newcommand{\Z}{\boldsymbol Z}

\author{Daphn\'e Aurouet\footnote{Corresponding author: Univ. Rennes, Ensai, CREST--UMR 9194, F-35000 France;\\ \indent \quad daphne.aurouet@ensai.fr} \qquad Valentin Patilea\footnote{Univ. Rennes,  Ensai, CREST--UMR 9194,  F-35000 France; valentin.patilea@ensai.fr}}

\date{\today}
\title{Discrete-time Markov chains with random observation times}

\begin{document}
	\maketitle
	
\begin{abstract}
We propose a new approach for estimating the finite dimensional transition matrix of a Markov chain using a large number of independent sample paths observed at random times.  The sample paths may be observed as few as two times, and the transitions are allowed to depend on covariates. Simple and  easy to update kernel estimates are proposed, and their uniform convergence rates are derived. Simulation experiments show that our estimation approach performs well. 
\end{abstract}

\medskip




\newpage \section{Introduction}

Discrete-time Markov chains are long-standing, easy-to-interpret probabilistic tools for analyzing sequence data. They are characterized by their transition matrix, which encapsulates the likelihood of each possible movement. They have found widespread application in fields such as biology, climate, ecology, engineering, finance, medicine, physics, and social studies. See, among many other references,  \cite{GN1962},  \cite{JLT2015}, \cite{AKK2021} and  textbooks  \cite{N1998}, \cite{BLM2005},  \cite{MM2015}, \cite{GS2020}. Markov chain models are also highly valued in mobility studies because the sequence records positions in a finite set of states at specific discrete time intervals.

Our study can be integrated into the broader application of discrete-time Markov chains to model the movement of statistical units, such as humans, animals, or physical goods such as banknotes. In human (e.g., \cite{WZLJ2021}) and animal mobility (e.g., \cite{MM2015}) studies, which have  benefited from the development of  \emph{Internet of Things} technologies, Markov chain modeling represents a natural and popular approach. Here, our focus is on cases where very large numbers of statistical units are available, as in banknote mobility. The subject was studied, for example, by \cite{M2021} who estimates the transitions of ten z\l{}oty banknotes between Polish regions and uses a log-linear model to incorporate economic and geographical factors to model the transition probabilities. \cite{BHG2006} propose to use banknote flows in the USA to study the spread of diseases. For the banknote They consider a continuous time random-walk process that incorporates scale-free jumps as well as long waiting times between displacements. Alternatively, we propose a discrete-time Markov chain approach to model the circulation of statistical entities, such as banknotes, because it produces interpretable results and allows covariates (also called predictors) to be introduced easily. 

We are also interested in a modeling approach that takes into account another feature of banknote circulation, that is  the difficulty in observing the statistical units frequently and/or regularly. The trajectory of a banknote during its lifetime is sparsely observed at random times. Partial observation scenarios occur in various contexts, such as systems that are prohibitively costly or infeasible to open for direct observation at each step. However, random observation patterns in Markov chain paths have not yet been extensively explored. \cite{MSW2008} consider Markov chains  with arbitrary state space observed at certain periodic or random time points. They address the question of which observation patterns allow us to identify the transition distribution. When identification is granted, they construct nonparametric estimators of functionals that depend on two successive observations of the chain. In \cite{BdCER2014}, the authors consider a homogeneous discrete-time Markov chain with finite state space and an incompletely observed trajectory. The purpose is to recover its transition matrix $P$. Assuming the time gaps between the observations are i.i.d., the observed process remains a homogeneous Markov chain with transition matrix $Q$. Imposing a sparsity assumption on $P$ and knowing the location of its zero entries, they build an estimator of the relationship between $P$ and $Q$ and derive an estimator of $P$ from the empirical estimator of $Q$.

Our primary goal is to estimate the transition matrices (also called \emph{stochastic matrices}) from partially and randomly observed trajectories of  Markov chains with a finite number of states. The trajectory of each statistical unit is supposed to be generated by a homogeneous Markov chain, and its transition matrix $\boldsymbol P(\z)$  is allowed to depend on the value $\z$ of the vector $\Z$ of covariates that are assumed to be constant over time. First, for any positive integer $\ell$, we consider the conditional transition probabilities given the covariates and that the time gap between two consecutive observations of a statistical unit is equal to $\ell$ time units. For example, the time unit can be a day or week. Let $\boldsymbol A_\ell (\z)$ be the matrix of these conditional transitions with observation gap equal to $\ell$. Under suitable and realistic assumptions on the time gaps between two consecutive observations of the same statistical unit, $\boldsymbol A_\ell (\z)$ should coincide with the conditional transition matrix power $\boldsymbol P^\ell(\z)$. The entries of $\boldsymbol A_\ell (\z)$  can be estimated by simple locally weighted means. The localization with respect to the continuous covariates is realized using a kernel and a suitable bandwidth, while simple indicator functions are used to stratify with respect to the discrete covariates. An estimate of $\boldsymbol P(\z)$ can then be obtained as the $\ell-$root of the estimate of $\boldsymbol A_\ell (\z)$. Estimates of $\boldsymbol A_\ell (\z)$ can be computed for the most frequent values for the interval between two consecutive observations of a statistical unit. After taking the corresponding matrix roots, this leads to several estimates for  $\boldsymbol P(\z)$ which can be averaged to yield a more accurate 
 estimator of the conditional  transition matrix given $\Z=\z$. The trajectory of a statistical unit can be observed on as few as two occasions. In contrast to previous studies on incompletely observed Markov chain paths, we leverage information from large number of trajectories.

The simple roadmap for estimating the conditional transitions described above presents the problem of calculating the $\ell$-th root of the estimates of $\boldsymbol A_\ell (\z)$. In fact, several inconvenient cases can occur~: the $\ell$-th root may not exist, it may exist and not be stochastic, or there can be more than one stochastic $\ell$-th root. Recall that a stochastic (or transition) matrix has nonnegative entries, with each row summing to 1. Several studies have addressed the problem of the existence, number and  nature of roots of a stochastic matrix. For example, \cite{HL2011} identify two classes of stochastic matrices that have stochastic $\ell$-th roots for all $\ell$. Moreover, they derive necessary conditions that the spectrum of a stochastic matrix must satisfy in order for the matrix to have a stochastic $\ell$-th root.

The problem of the existence of a stochastic $\ell$-th root of a stochastic matrix is related to the so-called \emph{embedding problem} for Markov chains (see, e.g., \cite{K1962}, \cite{H2008}, \cite{EBD2010} \cite{VB2018}, \cite{CFSRL2023}). A Markov chain with transition matrix $P$ is embeddable (or the transition matrix $P$ is embeddable) if and only if there exists a matrix $G$ with nonnegative off-diagonal entries whose row sums all vanish, such that $P = \exp(G)$. The matrix $G$ is usually called a \emph{generator} matrix, and alternative terminologies are intensity or rate matrix, skeleton, or $Q-$matrix. \cite{K1962} has shown that a Markov chain is embeddable if and only if it is nonsingular and has stochastic matrix roots of arbitrary order. In fact, if a Markov chain is embeddable with generator $G$, then  $\exp(G/\ell)$ is a stochastic $\ell$-th root of its transition matrix, for any $\ell$. Moreover, any positive power of an embeddable matrix is an embeddable matrix. The characterization of the set of embeddable matrices is an active research field. See, for example, \cite{EBD2010},  \cite{VB2018}, \cite{CFSRL2023} for some recent research. In particular, embeddability has certain consequences on the spectrum of a transition matrix, and the set of embeddable Markov matrices with distinct eigenvalues is relatively open and dense in the set of all embeddable Markov matrices. See, e.g., \cite{EBD2010}. Another aspect of the embedding problem is that an embeddable matrix does not necessarily possess a unique generator. This identifibiality problem, related to the  uniqueness of the real logarithm of a matrix,  has been solved in some particular cases. For example, \cite[Theorem 3]{C1972} shows that a diagonally dominant embeddable matrix has only one generator. Additional conditions were provided in \cite{C1973}. Recently, \cite{CFSRL2023} derived bounds for the number of Markov generators when Markov matrices have different eigenvalues.

While significant progress has been made to  understand the problems of the existence and uniqueness of the roots of a stochastic matrix that are also stochastic, statistical analysis faces additional challenges. Even if the modeling paradigm based on Markov chains is realistic and the theoretical transition matrix is embeddable, there is no guarantee that the estimates of the powers of the transition matrix admit roots that are stochastic. See, for example, \cite{BS2005} for a discussion on the case of the classical maximum likelihood estimation of the generator for a discretely observed Markov jump process with finite state. In the context of datasets arriving in streams, as it is the case in the study of banknote mobility, there is no guarantee that an estimate of a power of the transition matrix that admits roots that are stochastic will preserve these properties after being updated. Furthermore, the embeddability property may fail, or fail after updating the estimate . 
Therefore, we follow a pragmatic approach through regularization. 

There are several quantities  for which the regularization can be considered in other context, such as the roots of $\boldsymbol A_\ell (\z)$ for making them stochastic, or  the logarithm of  $\boldsymbol A_\ell (\z)$ for making it a generator matrix. 
If a $\ell$-th matrix root of the estimate of $\boldsymbol A_\ell (\z)$ exists but it is not stochastic, then one can search an approximation which is a stochastic matrix; see \cite{DM2024} for a recent reference. Alternatively, several regularization ideas for generator matrices have been proposed,  \cite{IRW2001}, \cite{KS2001}, \cite{EBD2010}. Because it looks for an optimal approximation of the estimates by matrices belonging to a space where  the requirements for stochastic or generator matrices are satisfied, regularization generally involves numerically complex minimization of a loss function. For simplicity and because the empirical studies reported in the literature, as well as our own extensive numerical investigations, suggest  that it performs well, we  consider here the generator regularization approach. For this, we  use the simple ideas described in \citep[Section 3]{IRW2001}, see also \cite{KS2001}, based on the remark that the identity $\boldsymbol P (\z)=\exp (\ell^{-1}\log(\boldsymbol A_\ell (\z)))$ should hold true, provided $\boldsymbol A_\ell (\z)= \boldsymbol P^\ell (\z)$. The idea is then to check whether the logarithm of the estimated $\boldsymbol A_\ell (\z)$ exists and satisfies the conditions of a generator matrix, and if this is not the case, to suitably modify the estimate of $\boldsymbol A_\ell (\z)$ or the logarithm of this estimate, respectively.

The remainder of this paper is organized as follows. Section \ref{sec:chap2-model-methodo} presents the Markov chain model and  assumptions regarding the distribution of gaps between consecutive observations of a trajectory.
Our modeling assumptions guarantee that the distribution of the observations, obtained from independent, partially observed Markov chain sample paths, can be characterized by the powers of the conditional transition matrix given the value of the covariates. The data can then be used to estimate the powers of the conditional transition matrix. In Section \ref{sec:chap2-model-estimation}, we introduce the recursive kernel estimators of these powers. After regularization of the generator estimate, taking the roots of the estimated powers, we obtain several estimates of the  conditional transition matrix of the Markov chain given the covariates, which we finally aggregate in a more precise estimator. In Section \ref{sec:chap2-model-theory}, we derive the uniform consistency and convergence rates of the aggregated estimator.  In Section \ref{sec:chap2-emp-analysis}, we present the results on the finite sample performance of our estimator. 
The proofs and technical justifications are collected in the Appendix, where some facts about the existence of logarithms and the stochasticity of root matrices are also recalled. Details on the design of the simulation experiments are provided in the Supplementary Material.

\section{Methodology}\label{sec:chap2-model-methodo}

Our modeling approach is based on discrete-time Markov chains with a finite number of states but is observed at random times. The objective is to estimate the conditional transition matrix of the chain, given the value of the covariate vector, using a large set of independent sample paths of the chain. First, we need to make some mild assumptions about the chain and the distribution of the random gaps between the observation times of a trajectory. Then, the distribution of our observations can be characterized by the powers of the conditional transition matrix given the value of the covariate vector. We use the estimates of these powers, together with the  logarithm  and exponential matrix functions,  to construct an estimate for the conditional transition matrix of the chain. We also use the simple regularization ideas proposed in \cite{IRW2001}, \cite{KS2001}.

\subsection{Markov chains with random observation times}\label{subsec_Ma}

We consider  $X = (X_t)_{t \in\mathbb N}$  a discrete-time Markov chain of order 1, with values in the finite state space $\mathcal S = \{1, 2, \dots, S \}$. The necessary Markov chains theory for the following can be found in textbooks \cite{N1998} and \cite{GS2020}. A vector $\Z \in\mathcal Z \subset \mathbb R^{p}\times  \mathcal Z_d  $ contains the  discrete covariates supported on $\mathcal Z_d$ and the $p$ continuous covariates.  We suppose $\Z$ independent of the time $t$. In the following, $\z \in \mathcal Z $ is any value in the support $\mathcal Z$ of $\Z$. Thus, we are interested in the process $(X_t, \Z)_{t\in \mathbb N}$. We consider the setup where several independent paths of $X$ are generated, and  a random draw of $\Z$ is observed for each path.

In the present work, we impose the following assumptions on $X$ and its conditional distribution given the covariate vector values.

\medskip 

\begin{assumption}\label{assump1}
    For any $\z\in\mathcal Z$,  the discrete time process  $X$ is a homogeneous, irreducible, and aperiodic Markov chain of order 1 with conditional transition probabilities
     \begin{equation}
     p_{ij}(\z)	= \mathbb P(X_{t+1} = j \mid X_t = i, \Z = \z) , \qquad \forall t\in\mathbb N, \forall i,j\in \mathcal S.
     \end{equation} 
     The conditional transition matrix of the process $X$ given $\Z = \z$ is denoted by $\boldsymbol P(\z)$.
\end{assumption}

\medskip 

Our focus is on the situation where the independent paths of $X$ are incompletely observed, with respect to $t$. More precisely, a sample path is observed at random times $0\leq T_0 < T_1 < \ldots < T_k < T_{k+1} < \ldots$.

To model this setup, let  $\tau_k\in\mathbb N^*  $ be a  random variable that indicates the time between two consecutive observation times of a sample path of $X$, that is we define
\begin{equation}
	\tau_k= T_k - T_{k-1}, \quad  k \in\mathbb N^*.
\end{equation}
Naturally, we have that $T_{k}=\tau_{k}+\cdots\tau_1+T_0$.
The observations associated with one sample path of $X$ are summarized by $(Y_k,\tau_k, \Z)_{k \in\mathbb N^*}$, with 
\begin{equation}\label{def_Yk}
Y_k = X_{T_k}\in\mathcal S,\qquad \text{ and } \quad T_{k}\in \mathbb N^*.
\end{equation}
The purpose is to estimate the conditional transition matrix $\boldsymbol P(\z)$ by using independent sample of sequences $(Y_k,\tau_k, \Z)_{k \in\mathbb N^*}$. 
Let us first consider the following assumption. 

\medskip

\begin{assumption}\label{assump2}
The initial time is  $T_0=0$. Moreover, for any integer $k\geq 1$, $\z \in\mathcal Z$, $i_{k-1},\ldots,i_1\in\mathcal S$ and $\ell^\prime_k, \ell^\prime _{k-1},\ldots,\ell^\prime_1\in\mathbb N^*$, we have 
    \begin{multline}\label{ass1}
        \mathbb P (Y_{k+1} = j , \tau_{k+1} = \ell  \mid  Y_{k} = i, \tau_k= \ell_k^\prime,  Y_{k-1}=i_{k-1}, \tau_{k-1}=\ell^\prime_{k-1},\ldots Y_1=i_1,\tau_1=\ell^\prime_{1},   \Z=\z) \\
        = \mathbb P (Y_{k+1} = j , \tau_{k+1} = \ell \mid Y_{k} = i, \Z=\z), \qquad \forall i, j \in \mathcal S, \forall \ell\in\mathbb N^* . \quad 
    \end{multline}
\end{assumption}

\medskip

\begin{remark}
For simplicity, in Assumption \ref{assump2} we assume that $T_0$ is not random, and  without loss of generality we can set it to zero. In the case where $T_0$ is random, it suffices to consider that the conditional probabilities in Assumption \ref{assump2} are also conditional, given the finite value of $T_0$. This additional technical aspect is omitted in the following as it does not change the procedure and the results below.  
\end{remark}

\medskip

Assumption \ref{assump2} imposes conditions on the observed variables, and is reasonable for the application we consider. It imposes a conditional lack of memory condition on   the process $(Y_k,\tau_k)_{k \in\mathbb N^*}$ given $\Z=\z$. In particular, it implies that given $\Z=\z$, $(Y_k)_{k \in\mathbb N^*}$ is Markovian of order 1. The justification is provided in Section \ref{sec_Markov_Y} of the Appendix.
To complete the Markov property on the observed process, we impose an assumption that allows us to obtain its conditional transition matrix as a power of the transition matrix of $X$, if the homogeneity condition imposed by Assumption \ref{assump1} holds.

\medskip

\begin{assumption}\label{assump3}
	For any $\z\in\mathcal Z$, and $t,t^\prime \in\mathbb N$ with $t<t^\prime$, it holds 
	\begin{equation}\label{ass2-condi}
        \qquad \mathbb P (X_{T_{k+1}} = j \mid X_{T_k} = i,  T_{k+1}=t^\prime , T_k=t, \Z=\z) = 
        \left(\boldsymbol P^{t^\prime -t}(\z)\right)_{ij},\quad \forall k\in \mathbb N, \forall i,j\in\mathcal S,
    \end{equation}
    where $\boldsymbol P (\z) $ is the transition matrix in Assumption \ref{assump1} and, for any $\ell \geq 1$, $\boldsymbol P^\ell (\z) $ denotes the $\ell$-th power of $\boldsymbol P (\z) $.
\end{assumption}

\medskip

Our objective is to estimate  $\boldsymbol{P}(\z)$ using the available data, with the guarantees given by the assumptions. It worth noting that, for $\ell = t^\prime - t$ and any $k\geq 1$, using Assumption \ref{assump3} we have
\begin{equation}\label{power_Y1}
		\mathbb P (Y_{k+1} = j \mid Y_{k} = i,  \tau_{k+1} = \ell, \Z=\z) = \left(\boldsymbol P^{\ell}(\z)\right)_{ij},\qquad \forall i,j\in\mathcal S.
\end{equation}
Then, for $k\geq 1$, we have
\begin{equation}\label{eq:link-to-est}
	\mathbb P (Y_{k+1} = j, \tau_{k+1} = \ell \mid Y_{k} = i, \Z=\z) 
	= \left(\boldsymbol P^{\ell}(\z)\right)_{ij} \; \mathbb{P}(\tau_{k+1} = \ell \mid Y_{k} = i,  \Z=\z),
 \end{equation}
from which an expression of $\left(\boldsymbol P^{\ell}(\z)\right)_{ij}$ can be derived. 
To be convenient for the estimation of  $\boldsymbol{P}^\ell(\z)$ using the data containing independent sequences $(Y_k,\tau_k, \Z)_{k\geq 1}$, the probabilities of the left- and right-hand sides of identity \eqref{eq:link-to-est} should not depend on $k$. This is imposed in the following conditional stationarity assumption for the observed variables. 

\medskip

\begin{assumption}\label{assump4}
	For any $\z\in\mathcal Z$, the conditional distribution of $\tau_{k+1}$ given $Y_k$ and $\Z=\z$ is the same for all $k\in \mathbb N$.  
\end{assumption}

\medskip

\begin{remark}
	Assumptions \ref{assump2} to \ref{assump4} do not impose conditional independence conditions, given the covariates on the random times $(T_k)_{k \geq 1}$ or between $(T_k)_{k \geq 1}$ and $X$. If $(T_k)_{k \geq 1}$ and $X$ are conditionally independent given the covariates, then $(\tau_k)_{k \geq 1}$ is also conditionally independent of $X$ given the covariates and several of the imposed conditions are automatically satisfied,  as shown below. 
\end{remark} 

\medskip

\begin{lemma}\label{lemma_ind}
Assumption \ref{assump1} holds true. Assume $T_0$ is not random, and $(\tau_k)_{k\geq 1}$ is conditionally independent of $X$  given $\Z$. Furthermore, for any $\z\in\mathcal Z$, the positive, integer-valued variables $\tau_k$, $k\geq 1$, are conditionally independent and have the same conditional distribution given $\Z=\z$. Then,  Assumptions \ref{assump2} to \ref{assump4} also hold true.
\end{lemma}


\subsection{Roots of stochastic matrices}

Motivated by \eqref{eq:link-to-est}, for $\ell \geq 1$, let us consider the matrix $\boldsymbol A_{\ell}(\z)$ with the components 
\begin{equation}\label{MC_P_fund}
\left(\boldsymbol{A}_{\ell}(\z)\right)_{ij}= \frac{\mathbb P (Y_{k+1} = j, \tau_{k+1} = \ell \mid Y_{k} = i, \Z=\z) }{\mathbb P (\tau_{k+1} = \ell \mid Y_{k} = i, \Z=\z)} ,\quad i,j\in\mathcal S, \z\in\mathcal Z.
\end{equation}
It worth noting that $\boldsymbol A_{\ell}(\z) $ is a stochastic matrix that depends only on the distribution of the observed variables. If the Assumptions \ref{assump1} to \ref{assump4} hold true, than $\boldsymbol A_{\ell}(\z)$ does not depend on $k$, and is equal to the $\ell-$th power of the conditional transition matrix  $\boldsymbol P(\z)$. A natural idea is then to estimate $\boldsymbol A_{\ell}(\z) $ and next define an estimator of $\boldsymbol P (\z)$ as the $\ell$-th root of $\boldsymbol A_{\ell}(\z) $, provided that such matrix root exists and is a stochastic matrix. 

Informally, by extending the identities satisfied by positive real numbers, we can expect the following relationship~: for any $\z\in\mathcal Z$, 
\begin{equation}\label{eq:Px.a1}
 \boldsymbol P(\z) = \exp\left\{\frac{1}{\ell} \log\left( \boldsymbol A_{\ell}(\z)\right)\right\}.
\end{equation}
If such a relationship holds true, an estimator of the conditional transition matrix can be obtained by plugging into the equation an estimator of $ \boldsymbol A_{\ell}(\z)$.
The definitions of the matrix functions exponential and logarithm, as well as their basic properties of existence and uniqueness, are recalled  in Appendix \ref{app_exp_log}. While the model assumption guarantee \eqref{eq:link-to-est}, from which we get $\boldsymbol{A}_{\ell}(\z)=\boldsymbol{P}^{\ell}(\z)$, the existence of a real logarithm $\log\left( \boldsymbol A_{\ell}(\z)\right)$ is not automatically granted and we therefore impose the following convenient condition. For any sets $A,B$, $d_{\rm H}(A,B)$ denotes the Hausdorff distance between $A$ and $B$. Moreover,  $\mathbb R^{-}=(-\infty,0]$ and for any matrix $\boldsymbol A$, $\sigma[\boldsymbol A]\subset \mathbb C$ denotes its eigenvalue set (spectrum).

\medskip

\begin{assumption}\label{assump5}
A constant $c>0$ exists such that, for any $\z\in\mathcal Z$ and  any eigenvalue $\lambda(\z)$ of $\boldsymbol P(\z)$, it holds  $|\lambda(\z)| \geq c $. Moreover, there exists a set $\mathcal L\subset [2,L]$ of positive integers such that  for any $\ell\in\mathcal L$, $\inf_{\z\in\mathcal Z}d_{\rm H}(\sigma[\boldsymbol{A}_{\ell}(\z)],\mathbb R^{-})\geq c$.
\end{assumption}

\medskip

Assumption \ref{assump5} imposes that the eigenvalues of the matrix $\boldsymbol{A}_{\ell}(\z)$ stay away from the negative real half-line, uniformly with respect to $\z$. As a consequence , there is a unique real principal logarithm of  $\boldsymbol{A}_{\ell}(\z)$, $\forall \z$. Moreover, $\log (\boldsymbol{A}_{\ell}(\z)^{1/\ell})= \ell^{-1} \log (\boldsymbol{A}_{\ell}(\z))$; see \citep[Th. 1.31, 11.2]{H2008}, see also our Appendix \ref{app_exp_log}. Thus, \eqref{eq:Px.a1} can be used as an estimating equation for $\boldsymbol P(\z)$.


\section{Estimation}\label{sec:chap2-model-estimation}

Data are obtained from $N$ independent sample paths of  $(Y_k,\tau_k)_{k\geq 1}$, with $Y_k=X_{T_k}\in\mathcal S$, $\tau_k=T_k-T_{k-1}\in\mathbb N^*$,  and $N$ independent realizations of the covariate vector $\Z\in\mathcal Z$. For many applications, it is reasonable to consider that the independent paths are only observed at random time points in a fixed window, for example $[1,L]$ for some positive integer $L$. The theoretical investigation is also made easier using this approach.

\medskip

\begin{assumption}\label{assump_1to4}
Let $\mathcal Y = (Y_{k},\tau_{k},\Z_{k})_{k\geq 1}$ be a process defined as in Section \ref{subsec_Ma} such that Assumptions \ref{assump1} to \ref{assump4} hold true. Let $\mathcal Y_m$, $m\geq 1$, be independent copies of $\mathcal Y$. The data obtained from the $\mathcal Y_m$'s are composed of the vectors  $(Y_{m,k},\tau_{m,k},\Z_{m})$,  $1\leq k \leq M_{m}\leq L$, $1\leq m \leq N$, where $L\geq 1$ is a given integer. 
\end{assumption}
\medskip

Sample paths must have at least two data points, at $T_0=0$ and  another time smaller than or equal to $L$. The values $M_m$ are thus random but bounded by $L$.

Our estimation approach is based on the identity in \eqref{eq:Px.a1} and a version of the definition \eqref{MC_P_fund}. More precisely, by the definition of the conditional probability, for any $\z$ in the support of $\Z$ we have
$$
\mathbb P (Y_{k+1} = j, \tau_{k+1} = \ell \mid Y_{k} = i, \Z=\z) = \frac{\mathbb P (Y_{k+1} = j, \tau_{k+1} = \ell , Y_{k} = i \mid \Z=\z)f_{\Z}(\z)}{\mathbb P (  Y_{k} = i \mid  \Z=\z)f_{\Z}(\z)},
$$
and
$$
\mathbb P ( \tau_{k+1} = \ell \mid Y_{k} = i, \Z=\z) = \frac{\mathbb P ( \tau_{k+1} = \ell , Y_{k} = i \mid \Z=\z)f_{\Z}(\z)}{\mathbb P (  Y_{k} = i \mid  \Z=\z)f_{\Z}(\z)}.
$$
Here, $f_{\Z}(\z)$ is the density with respect to the product measure of the Lebesgue measure and the counting measure, assumed to exist. Therefore, the elements of the matrix $\widehat {\boldsymbol A}_{\ell}(\z)$ can equivalently be written as 
\begin{equation}\label{new_defAl}
	\left(\boldsymbol{A}_{\ell}(\z)\right)_{ij}= \frac{ \mathbb P (Y_{k+1} = j, \tau_{k+1} = \ell , Y_{k} = i \mid \Z=\z)f_{\Z}(\z) }{ \mathbb P ( \tau_{k+1} = \ell , Y_{k} = i \mid \Z=\z)f_{\Z}(\z) } ,\quad i,j\in\mathcal S, \z\in\Z.
\end{equation}

For each $\left(\boldsymbol{A}_{\ell}(\z)\right)_{ij}$ in \eqref{new_defAl}, we propose to estimate the  numerator and denominator using nonparametric smoothing with kernels. Let $\z= (\z_{c},\z_{d})\in\mathcal Z $ be a point in the support of $\Z=(\Z_c,\Z_d)$, where $\z_{c}=(\z_{1,c},\ldots,\z_{p,c})$ is a point in the support of $\Z_{c}=(\Z_{c,1},\ldots,\Z_{c,p})$ the vector of continuous components of $\Z$, and $\z_d$ is a point in the support $\mathcal Z_d $ of the discrete covariates $\Z_d$. Let $K(\cdot)$ be a nonnegative  univariate kernel with support $[-1,1]$, and $h$ a bandwidth. We use the following notation~: 
\begin{equation}
    \boldsymbol K_{ h} \left( \Z-\z \right) = h^{-p}   K\left( (\boldsymbol Z_{1,c} - \z_{1,c} )/h \right) \times \cdots\times  K\left( (\boldsymbol Z_{c,p} - \z_{c,p}) /h \right)\times \mathds{1} \left\{ \boldsymbol Z_{d} = \z_{d}\right\} ,  
\end{equation}
where $\mathds{1} \left\{\cdot \right\}$ is the indicator function.

Let $h_m$, $m\geq 1$, be a sequence of bandwidths. For any $1\leq \ell \leq L$ and $\z\in\mathcal Z$, the estimator of the element $(i,j)$ of the matrix $\boldsymbol A_\ell(\z)$  is
\begin{equation}\label{eq:Ax-condi}
    \left(\widehat{\boldsymbol A}_{\ell}(\z)\right)_{ij}  =
    \frac{
        \sum_{m=1}^N \sum_{k=1}^{M_{m}} \mathds{1} \left\{ Y_{m,k} = j, \tau_{m,k}=\ell, Y_ {m,k-1} = i \right\} \boldsymbol K_{ h_m} \left( \Z_{m} - \z  \right)
    }{
        \sum_{m=1}^N   \sum_{k=1}^{M_{m}}  \mathds{1} \left\{ \tau_{m,k}=\ell, Y_{m,k-1} = i \right\}  \boldsymbol K_{ h_m} \left( \Z _{m} - \z  \right)
    },\quad i,j\in\mathcal S. 
\end{equation}
(The rule $0/0 = 0$ applies.) The same bandwidth is used for observations from the same sample paths; however $h_m$  can vary from one trajectory to another. Using an updated bandwidth for each new sample path is required in the case where the estimates need to be continuously updated. This case is discussed in the following. 

\medskip

\begin{remark}\label{rec_Al}
Our estimation approach is designed for application to very large datasets, in particular for streaming data. In applications, a few more observations can be added to each sample path, and more importantly, many other sample paths can be observed. The numerator and denominator in \eqref{eq:Ax-condi} can be easily updated in such situations, with low memory resources and computational complexity. For simplicity, we  detail only the case in which the data are updated with observations from new sample paths. For each $\ell$, we can rewrite the estimator in \eqref{eq:Ax-condi} in the form
\begin{equation}\label{eq:Ax-condi_bis}
	\left(\widehat{\boldsymbol A}_{\ell}(\z)\right)_{ij}  = \frac{\widehat U_{T,N}(i,j;\z,\ell)}{\widehat U_{B,N}(i;\z,\ell)}, \qquad i,j\in\mathcal S,\z\in\mathcal Z,
\end{equation}
where
\begin{multline}\label{recur1}
\widehat U_{T,N}(i,j;\z,\ell ) = \frac{N-1}N \widehat U_{T,N-1}(i,j;\z,\ell) \\+   \frac{1}N \sum_{k=1}^{M_N} \mathds{1} \left\{ Y_{N,k} = j, \tau_{N,k}=\ell, Y_ {N,k-1} = i \right\} \boldsymbol K_{h_N} \left( \Z_{N} - \z  \right),
\end{multline}
and
\begin{equation}\label{recur2}
\widehat U_{B,N} (i;\z,\ell)= \frac{N-1}N\widehat U_{B,N-1}(i,j;\z,\ell) + \frac{1}N  \sum_{k=1}^{M_N} \mathds{1}\! \left\{   \tau_{N,k}=\ell, Y_ {N,k-1} = i \right\} \boldsymbol K_{h_N} \left( \Z_{N} - \z  \right).
\end{equation}
The recursions \eqref{recur1} and \eqref{recur2} can be simply initialized with $\widehat U_{T,0}(i,j;\z,\ell) = \widehat U_{B,0}(i,j;\z,\ell) $.
\end{remark}

\medskip

Given the estimate $\widehat {\boldsymbol A}_{\ell}(\z)$, whenever the matrix $\log(\widehat {\boldsymbol A}_{\ell}(\z))$ satisfies the conditions of a generator matrix,
\begin{equation}\label{eq:Px.l}
	\widehat {\boldsymbol P}_{\ell}(\z) = \exp\left\{\frac{1}{\ell} \log\left(\widehat {\boldsymbol A}_{\ell}(\z)\right)\right\},
\end{equation}
is a stochastic matrix. See, e.g., \citep[Th. 2.1.2]{N1998}. The matrix  $\widehat {\boldsymbol P}_\ell (\z)$  represents an estimator of the conditional transition matrix ${\boldsymbol P} (\z)$ given $\Z=\z$ based on the transitions observed after exactly $\ell$ periods of time. Recall that a \emph{generator} or \emph{$Q-$matrix} of finite dimension is a square matrix $Q$ with entries $q_{ij}$ characterized by the following conditions: (i) $q_{ii}\leq 0$ for all $i$; (ii) $q_{ij}\geq 0$ for all $i\neq j$; and (iii) $\sum_{j}q_{ij} = 0$ for all $i$. If $\log(\widehat {\boldsymbol A}_{\ell}(\z))$ exists and is not a generator matrix, we apply simple regularization to make it be such a matrix. Details are provided in Section \ref{sec_implemnt}. See also \cite{IRW2001}, \cite{KS2001}. 

To better exploit the information carried by the sample, we can consider several $\ell$ values, typically those with the largest frequencies. Let $[\underline L, \overline L]\subset[1,L]$ a range of $\ell$ and 
\begin{equation}\label{eq:Px-aggregated-emp-2}
	\widehat {\boldsymbol P}(\z) = \frac{1}{\sum_{\ell=\underline L}^{\overline L} 
		\widehat {\pi}_\ell(\z)}
	\sum_{\ell=\underline L}^{\overline L}	\widehat {\pi}_\ell(\z)  \widehat {\boldsymbol P}_{\ell}(\z),
\end{equation}
be an aggregated estimator of the conditional transition matrix ${\boldsymbol P} (\z)$ given $\Z=\z$. Here, $\widehat \pi _\ell(\z)$ is an estimator of  $\pi_\ell(\z)=\mathbb P(\tau_k = \ell \mid \Z=\z)f_{\Z}(\z)$. A simple estimator is obtained with 
\begin{equation}\label{nice_pi}
\frac{\widehat {\pi}_\ell(\z) }{\sum_{\ell=\underline L}^{\overline L} 
	\widehat {\pi}_\ell(\z) }=\frac{\sum_{i\in\mathcal S} \widehat U_{B,N} (i;\z,\ell)}{\sum_{i\in\mathcal S,\underline L\leq \ell\leq \overline L}\widehat U_{B,N} (i;\z,\ell)}.
\end{equation}
The range $[\underline L, \overline L]$ is determined by the practitioner based on the available data and the validity of the condition in Assumption \ref{assump5}. 

\subsection{The generator matrix regularization}\label{sub_sec_Reg}

If Assumptions \ref{assump1} to \ref{assump5} hold true, then  $\log({\boldsymbol A}_{\ell}(\z))$ is well-defined and is a generator matrix. With sufficiently many observations, it is reasonable to expect that the real logarithm of $\widehat {\boldsymbol A}_{\ell}(\z)$ exists. This can be easily seen from the integral representation of the logarithm of the invertible matrices $\boldsymbol A$ with no negative eigenvalues~: with $\boldsymbol I$ denoting the identity matrix, 
\begin{equation}\label{topm}
 \log(\boldsymbol A) = (\boldsymbol A-\boldsymbol I )\int_0^1 [t(\boldsymbol A-\boldsymbol I ) +  \boldsymbol I]^{-1} dt.
\end{equation}
See, e.g., \citep[Th. 11.1]{H2008}. Assumption \ref{assump5} imposes conditions on the spectrum of ${\boldsymbol A}_{\ell}(\z)$ making the integrand in \eqref{topm} well defined and uniformly bounded if $\boldsymbol A$ stays close to ${\boldsymbol A}_{\ell}(\z)$. Since the entries of $\widehat {\boldsymbol A}_{\ell}(\z)$ will be shown to uniformly concentrate around that of $ {\boldsymbol A}_{\ell}(\z)$ with high probability tending to 1, one can expect the $\log({\boldsymbol A}_{\ell}(\z))$ to be well defined over $\mathcal Z$, at least with sufficiently large samples. This is also confirmed by our extensive simulations.

However, there is no guarantee that the logarithm of $\widehat {\boldsymbol A}_{\ell}(\z)$ is a generator (or $Q-$)matrix. We therefore consider a regularization approach for making $\log({\boldsymbol A}_{\ell}(\z))$ a generator matrix. See also \cite{IRW2001}, \cite{KS2001}. More precisely, we propose the following procedure, with two options for the second step. 

\medskip

\begin{algo}\label{algo:regularization} Let $\boldsymbol A$ be a square matrix for which a real logarithm $\boldsymbol B=\log({\boldsymbol A})$ exists. Let   
	$\boldsymbol B_{ij}$, with $ 1\leq i,j\leq S $, be the entries of $\boldsymbol B$, where $S$ is the cardinal of $\mathcal S$. 
	\begin{enumerate}
\item[(S1)] \emph{[Set negative non-diagonal elements to zero]} Redefine
$$
\widetilde {\boldsymbol B}_{ij} \leftarrow \max\{0, \widetilde {\boldsymbol B}_{ij}\},\qquad  \forall j\neq i .
$$
		\item[(S2)] \emph{[Diagonal adjustment]} Redefine
		$$
		\widetilde {\boldsymbol B}_{ii} \leftarrow - \sum_{j=1,j\neq i}^{S}\widetilde {\boldsymbol B}_{ij} .
		$$

				\item[($S2^\prime$)] \emph{[Weighted  adjustment]} Redefine
		$$
		\widetilde {\boldsymbol B}_{ij} \leftarrow \widetilde {\boldsymbol B}_{ij} - \left|\widetilde {\boldsymbol B}_{ij} \right|\frac{\sum_{j^\prime=1}^{S}\widetilde {\boldsymbol B}_{ij^\prime} }{\sum_{j^\prime=1}^{S}\left|\widetilde {\boldsymbol B}_{ij^\prime}\right| } .
		$$
	\end{enumerate}
\end{algo}

\medskip

We will show in the following that, under mild conditions,  the  regularization Algorithm for making the $\log({\boldsymbol A}_{\ell}(\z))$ a generator matrix does not affect negatively the convergence rate of the transition matrix estimator. 



\section{Theoretical grounds}\label{sec:chap2-model-theory}

We now study the consistency and convergence rates for the estimators proposed in Section \ref{sec:chap2-model-estimation}. The consistency of our estimators $\widehat {\boldsymbol P}_{\ell}(\z)$, and the aggregated version $\widehat {\boldsymbol P}(\z)$, can be derived from the analytic properties of the matrix functions exponential and logarithm, and the asymptotic behavior of the estimators $\widehat {\boldsymbol A}_{\ell}(\z)$ of the power matrix ${\boldsymbol A}_{\ell}(\z)$. First, we study the uniform convergence rate of these estimators. Next, using the analytic properties of the matrix logarithm function, we derive the rate of uniform convergence for the estimator of the generator matrix $\log ({\boldsymbol A}_{\ell}(\z))$. In particular, we prove that the regularization proposed in Section \ref{sub_sec_Reg} does not affect the convergence rate. Finally, we derive the uniform convergence rate for the aggregate estimator $\widehat {\boldsymbol P}(\z)$.

\subsection{Uniform consistency of the power matrix estimator} \label{th_gr1}
	
	It is worth noting that the following relationship holds true~:
\begin{equation}\label{all_related}
\widehat U_{B,N} (i;\z,\ell) = \sum_{j\in\mathcal S}\widehat U_{T,N} (i,j;\z,\ell).
\end{equation}
Since $\mathcal S$ is finite, the relationships \eqref{all_related} and  \eqref{nice_pi} indicate that for the convergence of $\widehat {\boldsymbol P}(\z)$ it suffices to focus the attention on the asymptotic behavior of $\widehat U_{T,N} (i,j;\z,\ell)$.

Let us consider a more general recursive version of the estimator $\widehat U_{T,N} (i,j;\z,\ell)$, which may be useful for better handling situations with  unbalanced transition frequencies. For simplicity, where there is no confusion, we omit the arguments $i,j,\ell$, and denote this more general version by  $\widehat U_{N} (\z)$, where
\begin{multline}\label{better_recU}
\!\!	\widehat U_{N} (\z) = 	\widehat U_{N} (\z;i,j,\ell)  =\frac{1}{\Omega_N}\sum_{m=1}^N w_{N,m}\sum_{k=1}^{M_m} \mathds{1} \left\{ Y_{m,k} = j, \tau_{m,k}=\ell, Y_ {m,k-1} = i \right\} \boldsymbol K_{h_m} \left( \Z_{m} - \z  \right)\\
	\quad \text{ where } \; w_{N,m}>0 \; \text{ and }\; \Omega_N = \sum_{m=1}^N w_{N,m}
	\rightarrow \infty.
\end{multline}
Different choices of the weights $w_{N,m} \Omega_N^{-1}$ lead to different types of recursive density estimators~:  $w_{N,m}=1$ leads to the  estimator in \eqref{recur1}, see also the estimators in \cite{WW1969}, \cite{Y1971}, whereas $ h_m \sim m^{-\alpha}$, setting $w_{N,m}=h_m^{p/2}$  leads to a version of the estimator in \cite{WD1979}, while the choice  $w_{N,m}=h_m^{p}$ yields the estimator studied by \cite{D2013}.

\medskip

\begin{assumption}\label{assump1_th}
	\begin{enumerate}[I)]
\item\label{mxia19} The support $\mathcal Z_d$ of the discrete predictor $\Z_d$ is finite, and the support $\mathcal Z_c$ of the continuous predictors $\Z_c$  is a bounded hyperrectangle in $\mathbb R^p$.

\item\label{mxia2}  The covariate vector $\Z$ admits a density  $f_{\Z}(\z)$ with respect to the product measure of the Lebesgue measure and the counting measure. For any $\z_d\in\Z_d$, the density $f_{\Z}(\z)$ admits  partial  derivatives of second order with respect to $\z_c$  that are uniformly continuous on  $\mathcal Z_c $.
		
	\item\label{mxia2b} For any $\z_d\in\mathcal Z_d$,  $ \PP (Y_{k} = j, \tau_{k}=\ell, Y_ {k-1} = i \mid \Z=\z)$ admits partial  derivatives of second order with respect to $\z_c$ that are uniformly continuous on  $\mathcal Z_c $.
		
	\item\label{mxia2c} For any $i\in\mathcal S$ and $1\leq \ell\leq L$,  $\inf_{\z\in\mathcal Z} \PP ( \tau_{k}=\ell, Y_ {k-1} = i \mid \Z=\z)f_{\Z}(\z) >0$.
		
	\item\label{mker1}   $\boldsymbol K(\cdot)$ is a product kernel obtained with a symmetric, Lipschitz continuous density $K(\cdot)$  supported on $[-1,1]$,    $\mu_2(\boldsymbol K)   = \int _{-1}^1 u^2   K (u) du$, $c_{\boldsymbol K}=\int_{-1}^1   K ^2  (u) du$. Moreover,  $ h_m =cm^{-\alpha}$, where $\alpha \in (0,1/p)$, and $c\in[\underline c, \overline c]$. The weights are $\omega_{N,m}: = w_{N,m} \Omega_N^{-1}$ are given by $w_{N,m} =m^{\beta}$, with $0\leq \beta \leq \alpha p<1$.  
	\end{enumerate}
\end{assumption}

\medskip

The conditions in Assumption \ref{assump1_th} are standard. The finite support of the discrete covariates is assumed for convenience. It could be relaxed at the expense of more complicated writings, but this does not affect the main findings below. The bandwidth condition in Assumption \ref{assump1_th}-\ref{mker1} makes $\widehat U_{N} (\z)$ to depend on $c$, and the uniform convergence with respect to $c$ will allow for data-driven bandwidths with a given decrease rate $\alpha$. The choice of  $\omega_{M,n}$  allows the theory to include some more general recursive estimators as discussed above.

First, we study the uniform rate of convergence  of the stochastic part of  $\widehat U_{N} (\z)$.

\medskip

\begin{proposition}\label{stoch1}
	If Assumptions \ref{assump_1to4}, \ref{assump1_th}-\ref{mxia19}, \ref{assump1_th}-\ref{mker1} hold true,   then
	\begin{equation}
		\max_{1\leq \ell\leq L}\max_{i,j\in\mathcal S}	\sup_{c\in[\underline c, \overline c]} \sup_{\z \in\mathcal Z} \left| \widehat U_{N} (\z)  - \EE \left[\widehat U_{N} (\z) \right] \right| =  O_\PP \left( N^{-(1-\alpha p)/2}\sqrt{\log N } \right).
	\end{equation}
\end{proposition}

\medskip

\begin{remark}
For simplicity, we consider that the sample paths are independent, which is a reasonable assumption for many applications. The uniform convergence in Proposition \ref{stoch1} can be obtained as soon as a suitable concentration inequality for sums of centered variables is available, particularly under suitable types of weak dependence between the sample paths. For example, a concentration inequality under $\varphi-$mixing dependence is used by \cite{WL2004} to derive strong uniform convergence of recursive estimators similar to $\widehat U_{N} (\z)$. Alternatively, the short-range dependence condition is considered by \cite{LXW2013} and their Rosenthal and Nagaev-type inequalities can be used. 
\end{remark}

\medskip

As a consequence of Proposition \eqref{stoch1} and the identities \eqref{nice_pi}, \eqref{all_related},  we have the following uniform consistency result. 

\medskip

\begin{corollary}\label{stoch1_cor}
	If the conditions of Proposition \ref{stoch1} and Assumption \	\ref{assump1_th}-\ref{mxia2c} hold true, then 
	\begin{equation}
\max_{1\leq \ell \leq L}\max_{i,j\in\mathcal S}	\sup_{c\in[\underline c, \overline c]} \sup_{\z \in\mathcal Z} \left| \left(\widehat{\boldsymbol A}_{\ell}(\z)\right)_{ij} - 
\frac{\EE \left[\widehat U_{T,N}(i,j;\z,\ell)\right] }{\EE \left[\widehat U_{B,N}(i;\z,\ell)\right]} \right|
=  O_\PP \left( N^{-(1-\alpha p)/2}\sqrt{\log N } \right).
	\end{equation}
	Moreover, 
	\begin{equation}
	\max_{1\leq \ell \leq L} 	\sup_{c\in[\underline c, \overline c]} 		\sup_{\z \in\mathcal Z} \left|  \widehat{\pi}_{\ell}(\z)  - \EE \left[\widehat{\pi}_{\ell}(\z)  \right] \right| =  O_\PP \left( N^{-(1-\alpha p)/2}\sqrt{\log N } \right).
	\end{equation}
	with $ \widehat{\pi}_{\ell}(\z)=\sum_{i\in\mathcal S} \widehat U_{B,N} (i;\z,\ell)/f_{\Z}(\z) $.
\end{corollary}

\medskip

Let $U(\z)$ be a short notation for $ \PP (Y_{k} = j, \tau_{k}=\ell, Y_ {k-1} = i \mid \Z=\z)f_{\Z}(\z)$.
We next study the bias of $\widehat U_{N} (\z)$ as defined in \eqref{better_recU}. For $\epsilon >0$, let  $\mathcal Z_{\epsilon}=\{\z\in\mathcal Z : \|\boldsymbol u - \z_c\|\geq \epsilon, \forall \boldsymbol  u \in \mathbb R^p\}$.

\medskip

\begin{proposition}\label{bias1}
	If   Assumptions \ref{assump_1to4} and \ref{assump1_th} hold true, then, for any $\epsilon>0$, 
	\begin{equation}
\max_{1\leq \ell\leq L}\max_{i,j\in\mathcal S}	\sup_{c\in[\underline c, \overline c]}    \sup_{\z\in  \mathcal Z_{\epsilon}}  \left|  \EE \left[\widehat U_{N} (\z)\right] - U(\z) \right| = O\left( N^{-2\alpha}\right),
	\end{equation}
and 	 $\max_{1\leq \ell\leq L}    \sup_{\z\in  \mathcal Z_{\epsilon}}  \left|  \EE \left[\widehat \pi_{\ell} (\z)\right] - \pi_\ell (\z) \right| = O\left( N^{-2\alpha}\right)$, with $\widehat \pi_{\ell} (\z) $ defined in Corollary \ref{stoch1_cor}.
\end{proposition}

\medskip

\begin{remark} 
	For a simpler presentation, Proposition \ref{bias1} and the results below exclude a small subset in the support of the continuous predictors to avoid bias increasing near the boundary of the support. The boundary problem can be avoided by considering  univariate boundary corrections and linearly combining estimates with several different bandwidths. Another possibility is to consider boundary modified kernels. See \cite{MS1999} and the references therein.   
 \end{remark}

\medskip

 Gathering the facts from above, we obtain the following result. The proof is now obvious and therefore omitted. 

\medskip

\begin{corollary}\label{cor_lambda1}
If the conditions of Proposition \ref{bias1} hold true, then for any $\epsilon>0$  and any matrix norm $\|\cdot\|$,
	$$
\max_{1\leq \ell \leq L} 	\sup_{c\in[\underline c, \overline c]} \sup_{\z \in\mathcal Z_\epsilon} \left\|  \widehat{\boldsymbol A}_{\ell}(\z)  - {\boldsymbol A}_{\ell}(\z)
\right\| =O_\PP \left(N^{-2\alpha}+ N^{-(1-\alpha p)/2}\sqrt{\log N }\right) .
	$$	
	The optimal exponent for the bandwidth is $\alpha = 1/(p+4)$, which gives the rate of uniform convergence $O_\PP\left(N^{-2/(p+4)}\sqrt{\log N} \right)$ for the transition matrices ${\boldsymbol A}_{\ell}$. The same rate of uniform convergence holds true for 
	$\widehat \pi_\ell$ defined in Corollary \ref{stoch1_cor}.
\end{corollary}

\medskip

\begin{remark}
	Concerning the pointwise quadratic risk  of the kernel estimator $\widehat U_N ( \z)$, by the proof of Proposition \ref{bias1}, 
	if $h_m= cm^{-\alpha}$ with  $\alpha = 1/(p+4)$, we get
	\begin{equation}\label{bias_lam1}
		\EE \big[\; \widehat U_N ( \z)\big] -  G(\z)  = c^2\frac{N^{-2\alpha}(1\!+\!\beta)}{2(1+\beta -2\alpha) }\mu_2(\boldsymbol  K) \operatorname{Trace}(\mathcal H_G(\z))
		\{1+o(1)\},
	\end{equation}
	where $\mathcal H_G(\z)$ is the Hessian of $G(\z)= \PP (Y_ {k} = j,  \tau_{k}=\ell, Y_ {k-1} = i \mid \Z=\z) f_{\Z}(\z )$ with respect to the continuous components of $\Z$.  On the other hand, from the proof of Proposition \ref{stoch1}, the pointwise variance of $\widehat U_N ( \z)$ is   
	\begin{equation}\label{var_lam1}
		\operatorname{Var}\big [ \widehat U_N ( \z) \big] \leq  \frac{L \|G\|_\infty c^p_{\boldsymbol K}}{N^{1 -\alpha p}} 
		\frac{ G(\z)} {c^p}  \frac{(1+\beta)^2}{1+\alpha p+2\beta} \{1+o(1)\},
	\end{equation}
	where $\|\cdot\|_\infty$ denotes the uniform norm. 	The optimal constant $c$ for the bandwidth $\mathfrak h(i)$ is  obtained by minimizing the leading terms in the squared bias plus the variance, \emph{i.e.}, with $p=d_c+1$ and $ \alpha= 1/(p+4)$,
	\begin{equation}\label{opt_ct_RCH}
	c_{\rm opt}(\beta;\alpha,\z)=\left[\frac{p(1+\beta-2\alpha)^2L c^p_{\boldsymbol K}G(\z)}{(1+\alpha p +2\beta)\{\mu_2(\boldsymbol K)\operatorname{Trace}(\mathcal H_G(\z))\}^2}\right]^\alpha.
	\end{equation}
	Thus, the optimal bandwidth for the estimator $\widehat U_N ( \z)$ is equal to the optimal bandwidth for the  non-recursive kernel estimator of $G(\z)$ multiplied by  
	\begin{equation}\label{Cr_const}
	c_r(\beta)  =   \left[\frac{\beta p+p+2 }{2(p+4)}\right]^{1/(p+4)}<1.
	\end{equation}
	See \cite{KX2019} for the case $\beta=0$. A constant $c_r(\beta)<1$  is not surprising as a larger bandwidth in the early recursions results in a large bias, which cannot be compensated for by the reduction in the variance.
	
	Hence, the bandwidth rule for $\widehat U_N(\z)$ can be derived from common rules.
\end{remark}


\subsection{Uniform consistency of the transition matrix estimator} \label{th_gr2}

In the following results, $ \widehat{\boldsymbol A}_{\ell}(\z) $ is defined as in \eqref{eq:Ax-condi} with $h_m=cm^{-\alpha}$ and let  $ \widehat{\boldsymbol B}_{\ell}(\z) $ be the regularized version of $\log( \widehat{\boldsymbol A}_{\ell}(\z)) $, whenever such a real logarithm of $ \widehat{\boldsymbol A}_{\ell}(\z) $ exists. 

\medskip

\begin{proposition}\label{reg_rate}
	The Assumptions \ref{assump5} to \ref{assump1_th} hold true, and $\alpha = 1/(p+4)$. Then, for any $\epsilon>0$,  and any matrix norm $\|\cdot\|$,
	$$
	\max_{1\leq \ell \leq L}	\sup_{c\in[\underline c, \overline c]} \sup_{\z \in\mathcal Z_\epsilon} \left\|  (1/\ell)\widehat{\boldsymbol B}_{\ell}(\z)  - \log({\boldsymbol P}(\z))
	\right\| =O_\PP\left(N^{-2/(p+4)}\sqrt{\log N} \right).
	$$	
\end{proposition}

\medskip

We now have all the elements to derive the uniform consistency of the conditional tradition matrix $\boldsymbol P(\z)$. Let $\underline L, \overline L\in[1,L]$ be given integers such that $\underline L < \overline L$.
\begin{equation}\label{final_t_P}
	\widehat {\boldsymbol P}(\z) = \frac{1}{\sum_{\ell=\underline L}^{\overline L} 
		\widehat {\pi}_\ell(\z)}
	\sum_{\ell=\underline L}^{\overline L}	\widehat {\pi}_\ell(\z)  \exp\left((1/\ell)\widehat{\boldsymbol B}_{\ell}(\z)\right),
\end{equation}
with   $\widehat \pi _\ell(\z)$ defined in \eqref{nice_pi}.

\medskip

\begin{corollary}\label{final_touch}
If the conditions of Proposition \ref{reg_rate}  hold true, then for any $\epsilon>0$, and any matrix norm $\|\cdot\|$,
	$$
	\sup_{c\in[\underline c, \overline c]} \sup_{\z \in\mathcal Z_\epsilon} \left\|  \widehat {\boldsymbol P}(\z) - {\boldsymbol P}(\z)
	\right\| =O_\PP\left(N^{-2/(p+4)}\sqrt{\log N} \right).
	$$	
\end{corollary}

\medskip

\begin{remark}
	It is easy to see that all the consistency results given in this section remain valid with $\widehat U_{N} (\z)$ defined as in \eqref{better_recU}
	when the same bandwidth $h_N\sim N^{-\alpha}$ for all observations, instead of a sequence $h_m\sim m^{-\alpha}$. In other words, our approach can be applied to a batch without recursions, can be used to update estimates when a new batch is available, or can be implemented in a fully recursive manner. Theoretical guarantees can be derived for each of these types of implementations. For simplicity, we only consider the case in which the estimates of $\widehat{\boldsymbol A}_{\ell}(\z)$ are updated after each sample path. 
\end{remark}

\section{Empirical analysis}\label{sec:chap2-emp-analysis}


In this section, we  present the results of our simulation experiments. Our goal is to estimate the (conditional) transition matrix $\boldsymbol P$ (resp. $\boldsymbol P(\z)$) associated with a finite Markov chain $(X_t)_{t\geq 1}$ (resp. $(X_t, \Z)_{t\geq 1}$) from the realizations of the process $(Y_k,\tau_k)_{k\geq 1}$ (resp. $(Y_k,\tau_k, \Z)_{k\geq 1}$), as described in Section \ref{subsec_Ma}.


\subsection{Simulation design}\label{sec:chap2-simu-implementation}

The simulations were performed for two state spaces, one with three states and the other with five states. 
The transition matrix of the process $X$ according to the definition of $\mathcal S$ is 

\begin{equation}\label{sim_matrixP1}
	\boldsymbol P = 10^{-2}\begin{bmatrix}
        94.0007 & 3.4412 & 2.5581 \\ 
        3.8810 & 92.5639 & 3.5551 \\ 
        0.3831 & 2.5038 & 97.1131 \\	
	\end{bmatrix},
\end{equation}
and 
\begin{equation}\label{sim_matrixP2}
	\boldsymbol P = 10^{-2}\begin{bmatrix}
    	91.4828 & 1.7832 & 1.5797 & 3.9951 & 1.1592 \\ 
        0.4332 & 94.0624 & 3.5217 & 0.1473 & 1.8354 \\ 
        0.8712 & 1.7389 & 93.1986 & 1.1289 & 3.0624 \\ 
        0.3389 & 3.0794 & 2.7967 & 90.3348 & 3.4502 \\ 
        0.3325 & 3.7597 & 4.3798 & 2.8478 & 88.6802 \\ 
	\end{bmatrix}.	
\end{equation}
respectively. These matrices, as many others that we considered in simulations for which the results are not reported here, are generated randomly (details are given in the Supplementary Material). They admit a unique generator, as they are inverses of $M-$matrices and their diagonals are larger than 1/2. See Appendix \ref{M_matrix_sto} for details and references on the existence and uniqueness of the generator matrices.  


To incorporate the effect of covariates into the transition probabilities, we consider the following expression for the elements of the conditional transition matrix given $\Z=\z$,

\begin{equation}\label{eq:sim-link-func}
\left(\boldsymbol P(\z)\right)_{ij} = 
\frac{
	\exp\left\{ p_{ij} \; \psi(\z) \right\}
}{
	\sum_{v=1}^{|\mathcal S|} \exp\left\{p_{iv}\; \psi(\z) \right\}
}, \qquad \forall i,j \in \mathcal{S}, \quad \z \in \mathcal Z,
\end{equation}
where $p_{ij}$ are the elements of matrix $ \boldsymbol P$ and $\psi(\z)$ is a  function  of the covariates. It is worth noting that, in general, the existence of the $\ell$-th roots and their stochasticity cannot be ensured for all $\z \in \mathcal Z$, and some regularizing steps are sometimes required. The covariate vector $\Z$ has two components that are drawn independently. One continuous variable, named $\Z_{c}$, such that $\Z_{c}-1$ has a Beta distribution $\mathcal B(a,b)$ with parameters $a=b=2$. 
One discrete variable, $\boldsymbol Z_{d}$ is a Bernoulli variable with parameter $0.7$.
Finally, we consider $\psi(\Z_c, \Z_d) = 3 \Z_c (1.2 \Z_d + 0.8 (1-\Z_d))$ in our results.






To generate $N$ independent trajectories, for each $1\leq m \leq N$, we draw the initial value $Y_{m,1}$ from a discrete uniform distribution over $\mathcal S$ and the vector of covariates $\Z_m$. Next, given $Y_{m,k}$, $k\geq 1$, the  number $\tau_{m,k+1}$ of units of time before the next observation of the sample path is drawn from a Poisson distribution with parameter $\lambda$, to which we add 1. We let the parameter $\lambda$ of the Poisson distribution depend on the position $Y_{m,k}$. It is set to 10 if $Y_{m,k} \in \{ 1, 2\} $ and  to 15 if $Y_{m,k}\in \{3, \cdots, S\}$. Finally, given $\Z_{m} = \z$, $\tau_{m,k}=\ell$ and $Y_{m,k}=i$, we draw $Y_{m,k+1}$ from a multinomial distribution with the vector of parameters equal to the $i-$th row of the matrix $\boldsymbol{P}^\ell(\z)$. The steps are repeated $M_m$ times, where $M_m$ is the smallest integer such that $\sum_{k=1}^{M_m} \tau_{m,k} \geq L $. We set $L \in \{20, 40, 60\}$ allowing a longer trajectory as $L$ increases.
In the case without covariates, depending on the number of states, the matrix $\boldsymbol P$ in \eqref{sim_matrixP1} or \eqref{sim_matrixP2} is used instead to generate the sample paths of $X$. 

\subsubsection{Implementation aspects}\label{sec_implemnt}
 
The estimator of the conditional transition matrix $\boldsymbol P(\z)$ is computed as presented in Section \ref{sec:chap2-model-estimation}. The principal logarithms of the stochastic  matrices $\widehat{\boldsymbol{A}}_{\ell}(\z)$ as defined in \eqref{eq:Ax-condi} are computed using the function \texttt{logm} from the  \texttt{R} package \texttt{expm}, see \cite{expm2024}. The \texttt{logm} proposes two methods, "Higham08" and "Eigen". As indicated in \cite{expm2024}, the method "Higham08" based on the inverse scaling and squaring method with Schur decomposition, and is designed for non-diagonalizable matrices, see \citep[Algorithm 11.9]{H2008}. The method "Eigen" which tries to diagonalize the matrix for which the logarithm has to be computed. Nevertheless, the \texttt{logm} function may fail to provide results with both methods, and this occurs with small samples. Without any regularization on the $\widehat{\boldsymbol{A}}_{\ell}(\z)$ such that the \texttt{logm} function can deliver a result, our approach requires as many as tens of thousands of sample paths of lengths as described below. The details are provided in the Supplementary Material. Keeping in mind applications where the number of  sample paths is very large, we focus on the next problem, that is, when the result of the  \texttt{logm} function is not a generator matrix. 
In such cases, regularization via weighted adjustments, as presented in Algorithm \ref{algo:regularization}, is applied. 
Details on how often regularization is required are provided in the Supplementary Material. The same steps for computing the estimator of the transition matrix are used in the case without covariates (unconditional), where the target is the matrix $\boldsymbol P$ in \eqref{sim_matrixP1} or \eqref{sim_matrixP2}. In this case, the estimator $\widehat{\boldsymbol P}$ is computed using formulae \eqref{eq:Ax-condi}, 
\eqref{recur1}, \eqref{recur2}, and \eqref{nice_pi} after removing  the kernel factors.

For kernel smoothing, we consider a standard Gaussian density, which in practice can be considered to have a compact support. 
 As we choose a Gaussian kernel and $p=1$, we make the simple choice of $h_m =C   \widehat{\sigma}_{Z_c} m^{-1/5} $ where $\widehat{\sigma}_{Z_c}$ is an estimator of the standard deviation of the continuous predictor and $C$ is a constant that we let depend on the number of states $S$.  The method is evaluated on a grid of 4 points by considering the median value $\z_c=1.5$ and $\z_c = 1.7$ from the support of the distribution for the continuous variable, and the support of the discrete variable. The values for $\underline L$ and $\overline{L}$ are 6 and 20, respectively.
To evaluate the performance of the estimator, we calculate the error matrix $\boldsymbol U(\z) = \widehat {\boldsymbol P}(\z) -\boldsymbol P (\z)$ and use its spectral norm, denoted $\|\cdot\|_2$. In the unconditional  case (without covariates) the error matrix is $\boldsymbol U = \widehat {\boldsymbol P} -\boldsymbol P $. We consider sample sizes $N \in\{  5000, 10000, 20000, 40000\}$ and each experiment is replicated 500 times. Sample sizes are expressed in terms of individuals in the sample, which can be associated with one  or more returns. 

\subsubsection{Simulation results}

Figure \ref{fig:simu-line-median-spec-uncond} shows the spectral norm of the error matrix $\boldsymbol U  $ on a logarithmic scale when the data are generated with $S \in \{3,5\}$ and $L\in \{20, 40, 60\}$ without covariates. As expected, the performance improves with the sample size in terms of the number of points on the trajectory (set by $L$) and the number of trajectories in the sample (set by $N$). The latter has a more prominent effect. 
Conversely, the performance deteriorates with the number of states in the model, as the complexity increases. With three states, estimating the transition probabilities involves determining six free parameters, considering that the rows should sum to 1. Similarly, 20 free parameters must be estimated for a chain with five states. 
The effect of letting the distribution of the variable $\tau_{m,k+1}$ depend  on $Y_{m,k}$ can be seen in the variance of the error for the estimation of the transitions from the states $\{3,..., S\}$, which are then larger than those for the estimation of the transition from the states $\{1,2\}$. This pattern, which we notice in the results that are not reported, can be explained by the fact that being in the states $\{3,..., S\}$ implies longer expected return times and thus fewer observed transitions.

\begin{figure}[t]
	\centering
	\caption{\small The median (on the log-scale) of the spectral norm of the error matrix $\boldsymbol U $ for the transition matrix  without covariates per sample sizes over 500 replications. 
	}
	\includegraphics[width=0.8\linewidth, height=9cm]{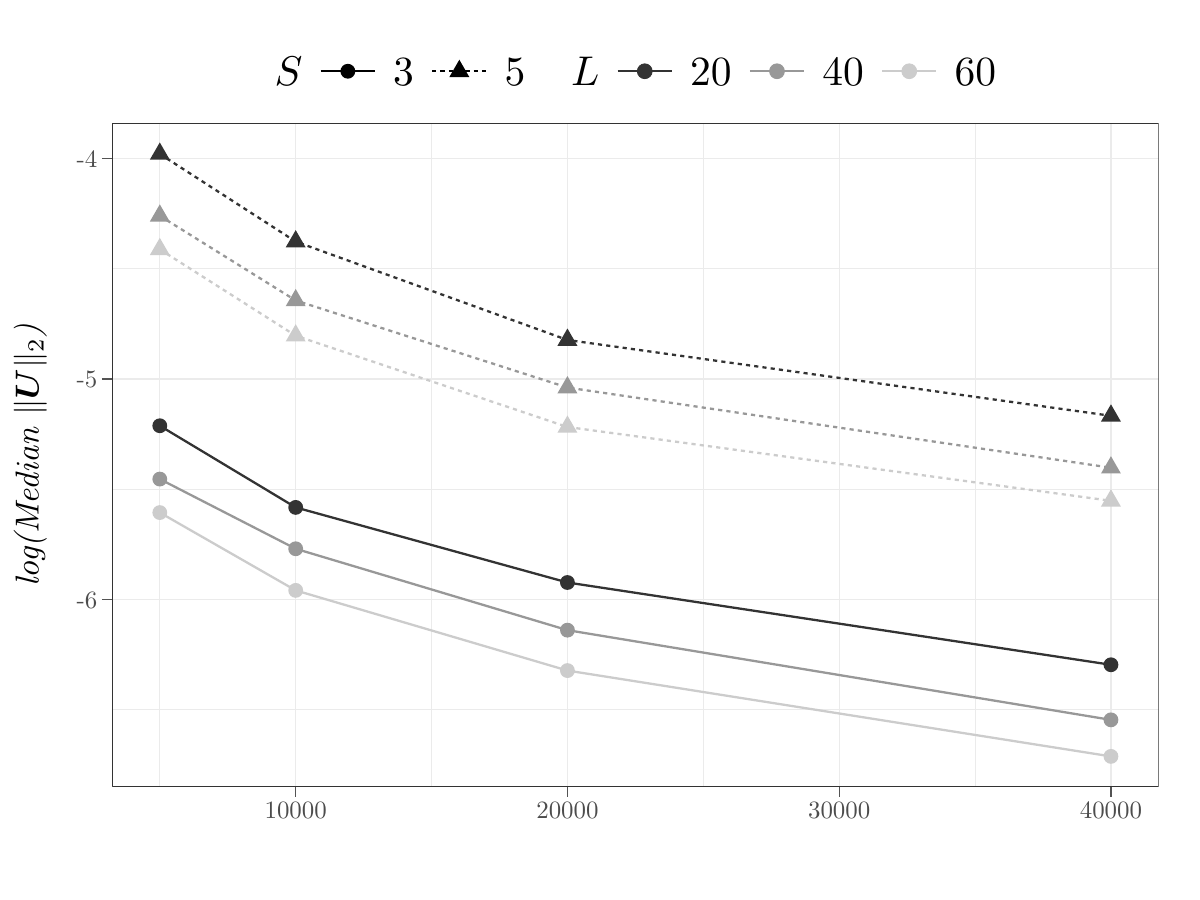}
	\label{fig:simu-line-median-spec-uncond}
	\vspace{-1cm}
\end{figure}

Figure \ref{fig:simu-line-median-spec-cond} displays similar results for $\boldsymbol U(\z)$ in the model with covariates. The lines represent  different combinations of covariates. 
 The performance depends on the values of $\z_c$ and $\z_d$. For the discrete variable, the choice of parameters for the Bernoulli random variable is 0.7, there are far fewer data points available for $\z_d=0$, and the accuracy is lower in this case. Additionally, the increase in performance with higher values of the continuous covariate is likely a consequence of the setup  \eqref{eq:sim-link-func}, which makes the diagonal of the true matrix more dominant with $\z_c=1.7$, and our method performs better in such situations. Nevertheless, the convergence of the estimator becomes slow for large state spaces and $\z_d$ values at zero. 

\begin{figure}[ht!]
	\centering
	\caption{\small  The median (on the log-scale) of the spectral norm  of the error matrix $\boldsymbol U (\z)$ for the transition matrix with covariates per sample sizes over 500 replications.
	}
	\includegraphics[width=0.95\linewidth, height=11.5cm]{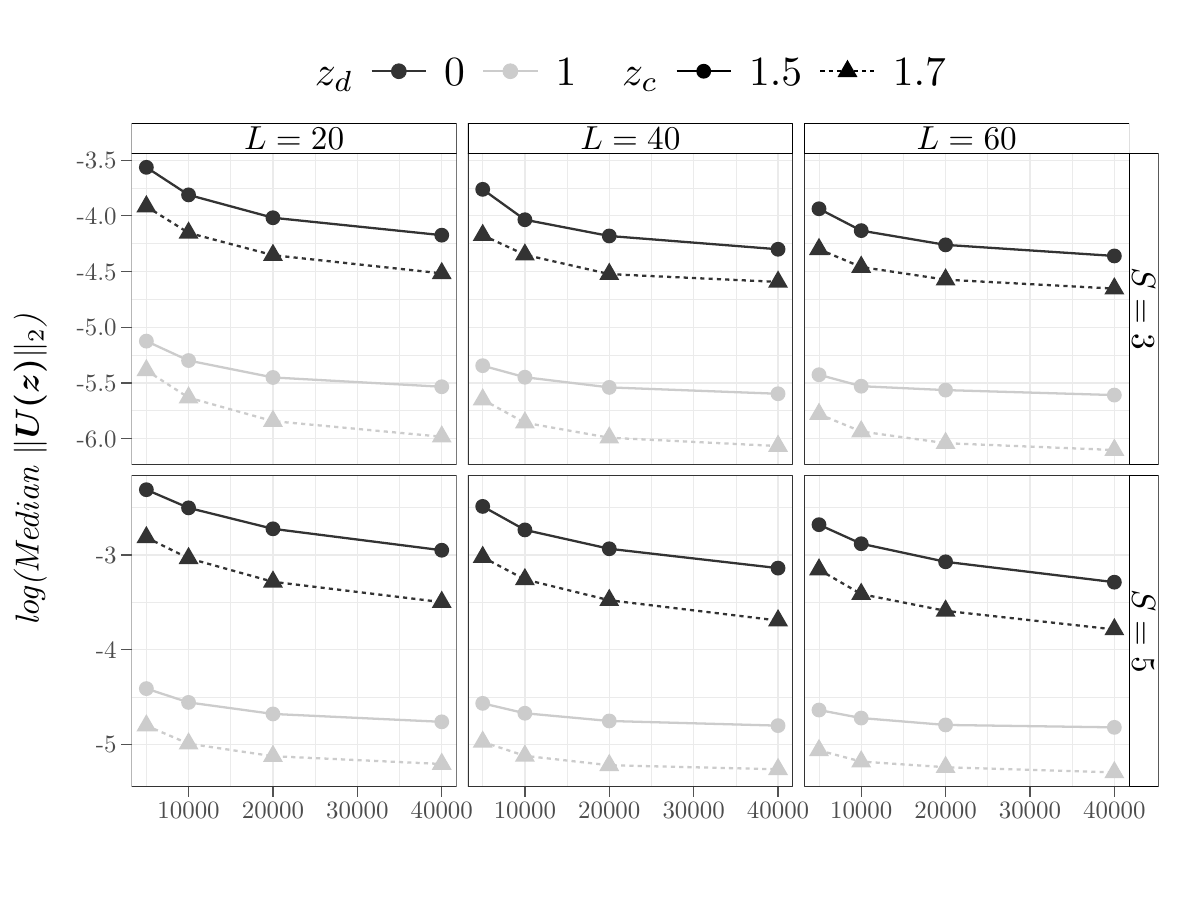}
	\label{fig:simu-line-median-spec-cond}
\end{figure}

\begin{figure}[ht!]
	\centering
	\caption{\small Boxplots of the spectral norm (on the log-scale) of the error matrix $\boldsymbol U (\z)$ for 500 replications over different sample sizes. From left to right, the panels display different sizes of the state space. Top panels are for the case with covariates at different points of evaluation, and the last for the model without.}
	\includegraphics[width=.95\linewidth, height=14cm]{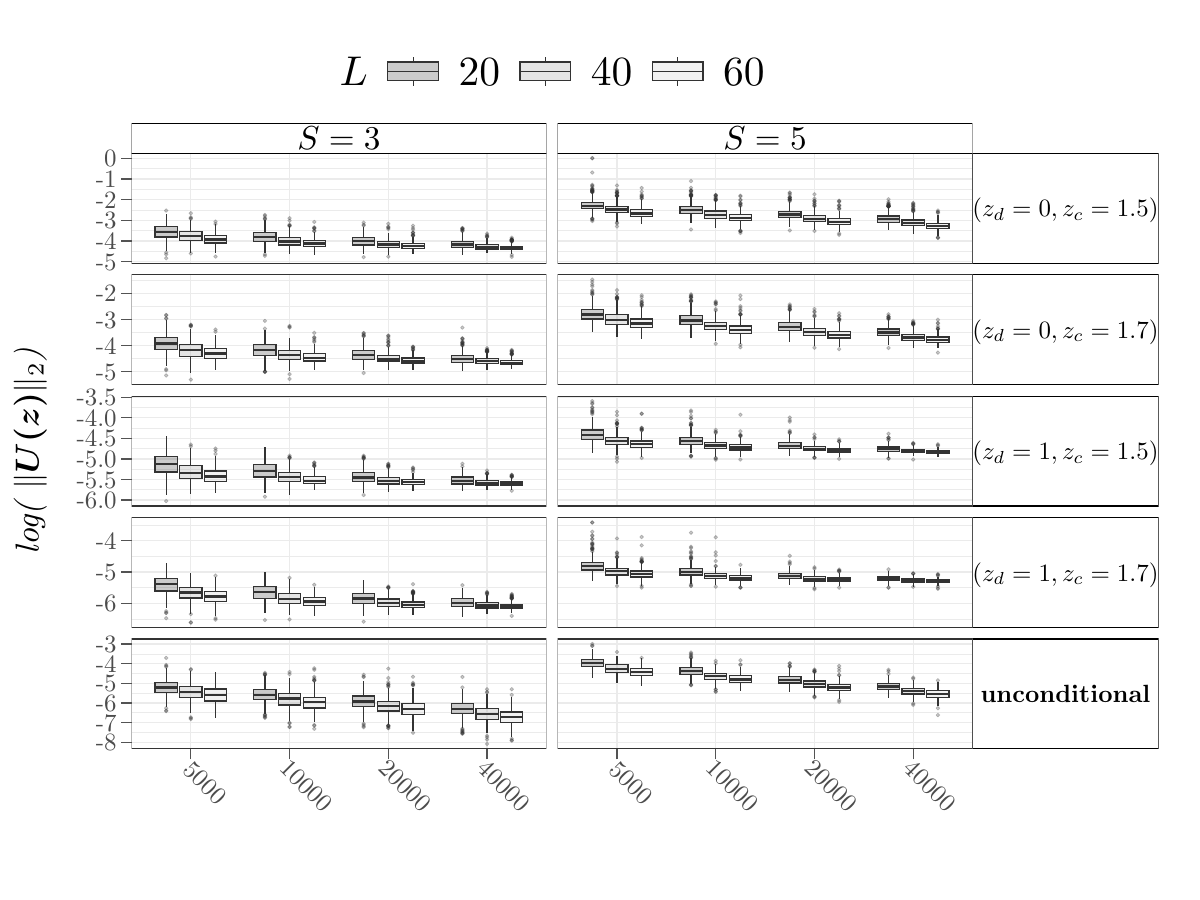}
	\label{fig:sim-boxplot-log-sec-norm}
	\vspace{-1cm}
\end{figure}

Figure \ref{fig:sim-boxplot-log-sec-norm} provides details of the spread of the spectral norm of the error matrix for different setups. The nonparametric approach for locally smoothing the effect of covariates appears appropriate, as the reduction in performance compared to the model without covariates, is marginal compared to the advantages gained in interpretation and flexibility. However, in practice, it is recommended to ensure that the transition matrix has a dominant diagonal. It is clear that the convergence rate is not uniform across the evaluation points. Better performances and rates are expected for more diagonally dominant transition matrices. Nonetheless, the number of states is the main reason for the performance decay. The Supplementary Material presents the results of the loss of performance with a state space of nine elements.


\section{Conclusions}\label{sec:chap2-ccl}

We introduced a novel approach for modeling the conditional transition matrices of a discrete-time Markov chain with a finite state space. The main challenge is of a statistical nature. The data contain a large number of independent sample paths, but they are observed at random times, as few as two times. Using a set of assumptions that seem reasonable for applications in the circulation, of humans, animals, or goods, we propose a nonparametric approach for the estimation of a  conditional transition matrix given the predictors. The predictors are constant over time and characterize a sample path, which is our statistical unit. The main idea is that the conditional transition matrix of interest is the $\ell-$root of conditional transition matrix after $\ell$ periods of time, which can be easily computed from the data, either in batch or recursively. The effect of the predictors is considered using kernel smoothing. Thus, taking the exponential of $\ell^{-1}$ times the logarithm of the conditional transition matrices after $\ell$ periods, estimated from the data, is expected to  yield an estimate of the conditional transition matrix of interest. Some regularization methods are necessary to make this simple idea implementable, and here, we propose to do it simply on the logarithms of the estimated transitions. 

The rates of uniform convergence of our matrix estimators are derived. The optimal rate for the bandwidth in kernel smoothing is calculated and related to the standard (batch) optimal kernel regression  bandwidth, when our estimates are calculated recursively. Simulation experiments revealed that the conditional transition matrix of the Markov chain can be estimated accurately and with little computing resources using existing \texttt{R} software, even if the number of sample paths is large.

\appendix
\section{Appendix} 


In this appendix, we recall some theoretical facts for finite-dimensional Markov transition matrices and provide proofs for the statements in Section \ref{sec:chap2-model-theory}. In Section \ref{app_exp_log}, the definitions and basic properties of exponential and logarithmic functions on matrices are presented. Some properties of the embeddable Markov matrices are recalled in Section \ref{M_matrix_sto}. In Sections \ref{sec_Markov_Y} and \ref{add_prooofs} we provide proofs. 


\subsection{Exponential and logarithm of matrices}\label{app_exp_log}

On the real line, the exponential can be defined by the limit of a series. This definition by series expansion can be easily extended to matrices. More precisely, for any square (finite-dimensional) matrix $\boldsymbol{Q}$, the exponential of $\boldsymbol{Q}$ is defined as 
$$
\exp(\boldsymbol{Q}) = \sum_{k=0}^\infty \frac{\boldsymbol{Q}^k}{k!},
$$
with $\boldsymbol{Q}^0 = \boldsymbol I$ and $\boldsymbol I$ the identity matrix. The convergence of the series can be understood as component-wise or using any matrix norm, and holds for any $ \boldsymbol{Q}$.  If $\boldsymbol{Q}_1$ and $\boldsymbol{Q}_2$ commute, then $\exp(\boldsymbol{Q}_1+\boldsymbol{Q}_2) = \exp(\boldsymbol{Q}_1)\exp(\boldsymbol{Q}_2).$
In particular, $(\exp(\boldsymbol{Q}/\ell))^\ell = \exp(\boldsymbol{Q})$, for all $\ell \geq 1$.
Moreover, if $ \boldsymbol{P}  = \exp(\boldsymbol{Q}) $, then for all $\ell \geq 1$,
\begin{equation}\label{eq:P_exp_A}
	\boldsymbol{P}^\ell   = \exp(   \ell \boldsymbol{Q} ) .
\end{equation}
Contrary to the exponential function of real numbers, the matrix exponential is not injective. This means that it can happen that $\exp(\boldsymbol{Q}_1)=\exp(\boldsymbol{Q}_2)$ even if $\boldsymbol{Q}_1\neq \boldsymbol{Q}_2 $.

The definition of the logarithm of a matrix is subtle. First, let us  recall the idea of defining the matrix function by the limit of a series. 
Whenever $\boldsymbol A= \boldsymbol I + \boldsymbol B$  with $\rho(\boldsymbol B)<1$, the logarithm of  $\boldsymbol A$ can be defined through the following series
\begin{equation}\label{series_logA}
\log(\boldsymbol A) = \log(\boldsymbol I+ \boldsymbol B) = \boldsymbol B-\frac{\boldsymbol B^2}{2}+\frac{\boldsymbol B^3}{3} -\frac{\boldsymbol B^4}{4} +\cdots,
\end{equation}
where $ \rho(\boldsymbol B)$ is the spectral radius (the largest eigenvalue in absolute value) of $\boldsymbol B$. 

The radius of convergence for the logarithm series is equal to 1, and it is easy to understand the impossibility of using the series expansion to define the logarithm of general matrices. To introduce the definition in full generality, let us recall that for $\boldsymbol A$ a $K\times K-$matrix with  eigenvalues $\lambda_1,\ldots,\lambda_L\in\mathbb C$, $L\leq K$, there exists a nonsingular matrix $\boldsymbol V$ such that 
$$
\boldsymbol A=\boldsymbol V \boldsymbol D \boldsymbol V^{-1}  \text{ where } \boldsymbol D= \begin{pmatrix}
	\boldsymbol D_1 & \cdots & 0 \\ 
	\vdots & \ddots & \vdots \\ 
	0      & \cdots & \boldsymbol D_L
\end{pmatrix} \text{ with } \boldsymbol D_l= \begin{pmatrix}
	\lambda_l & 1 & 0 & \cdots & 0 \\ 
	0 & \lambda_l & 1 & \cdots & 0\\
	\vdots &  \ddots & \ddots & \vdots & \vdots\\
	0 & 0 & \cdots & \lambda_l & 1\\ 
	0 & 0 & \cdots & \cdots & \lambda_l
\end{pmatrix},$$
and $\boldsymbol D_l$ is an $(r_l\times r_l)-$matrix  and $\sum_{1\leq l\leq L} r_l = K$. This decomposition of $\boldsymbol A$ is known as the \emph{Jordan canonical form}. The Jordan matrix $\boldsymbol D$ is unique up to the ordering of blocks $\boldsymbol D_l$, but the transforming matrix $\boldsymbol V$ is not unique. However, 
the definition yields an $f(\boldsymbol A)$ that can be shown to be independent of the particular Jordan canonical form used. See \cite[Section 1.2]{H2008}. Assuming the required regularity for a function $f$, the matrix $f(\boldsymbol A)$ is defined as
$$
f(\boldsymbol A)=\boldsymbol V f(\boldsymbol D) \boldsymbol V^{-1}  \text{ where }  f(\boldsymbol D)= \begin{pmatrix}
	f(\boldsymbol D_1)& \cdots & 0 \\ 
	\vdots & \ddots & \vdots \\ 
	0      & \cdots &f(\boldsymbol D_L)
\end{pmatrix} ,
$$
with 
$$
f(\boldsymbol D_l)= \begin{pmatrix}
	f(\lambda_l) & f^\prime (\lambda_l) & f^{\prime\prime} (\lambda_l)/2! & \cdots & f^{(r_l-1)} (\lambda_l)/(r_l-1)! \\ 
	0 & \lambda_l &  f^\prime (\lambda_l)  & \cdots & f^{(r_l-2)} (\lambda_l)/(r_l-2)! \\
	\vdots &  \ddots & \ddots & \vdots & \vdots\\
	0 & 0 & \cdots & \lambda_l &  f^\prime (\lambda_l)  \\ 
	0 & 0 & \cdots & \cdots & f(\lambda_l)
\end{pmatrix}. 
$$

The eigenvalues of $\boldsymbol A$ can be complex numbers and thus the logarithm of a matrix is not uniquely defined.
Indeed, a complex number is written $z=|z|e^{i\theta}$ where $\theta$ is the argument of $z$ (also denoted $\text{arg}(z)$). Since $e^{i\theta}=e^{i(\theta+2k\pi)}$, for any integer $k$, the argument of a complex number is not uniquely defined. The logarithm of $z\neq 0$ is then defined as $\log(z) = \log(|z|) + i \text{arg}(z)$. A branch of $\log(z) $ is defined by fixing the range of the imaginary part $\text{arg}(z)$, typically $(-\pi,\pi]$. The \emph{principal logarithm} (or principal branch of the logarithm) is the logarithm whose imaginary part lies in the interval $(-\pi,\pi]$. Moreover, the  derivatives of the logarithm are not defined at the origin; therefore the logarithm of matrices with zero eigenvalues cannot be defined. 

As pointed out by \cite[Section 1.4]{H2008}, one of the main uses of matrix functions is to solve nonlinear matrix equations,
such as $g(\boldsymbol X) = \boldsymbol A$, in particular the equation with $g(\cdot)$ the exponential. Any solution of the nonlinear equation  $\exp(\boldsymbol X)= \boldsymbol A$ is called a logarithm of $A$. However, for certain matrices $\boldsymbol A$, some of the solutions of $\exp(\boldsymbol X)=\boldsymbol A$  cannot be produced by the definitions of $f(\boldsymbol A)$ above, with $f(t) = \log(t)$. 
These solutions $\boldsymbol X$ are called \emph{nonprimary matrix functions}, whereas those obtained by the definition of $f(\boldsymbol A)$ are called \emph{primary matrix functions}.

The following properties are stated in \cite[Theorems  1.23, 1.31 and 11.2]{H2008}.
\begin{itemize}
	\item Existence of  real logarithm:
	\begin{enumerate}
		\item The nonsingular real matrix $\boldsymbol A$ has a real logarithm if and only if $\boldsymbol A$ has
		an even number of Jordan blocks of each size for every negative eigenvalue. 
		\item If $\boldsymbol A$ has negative eigenvalues then no primary 
		logarithm is real.
	\end{enumerate}
	\item If the real matrix $\boldsymbol A$ has no eigenvalues in $(-\infty,0]$, there exists a unique logarithm of $\boldsymbol A$ which is real.  This matrix, is then the principal logarithm of $\boldsymbol A$. In particular, if $\rho(\boldsymbol A)<1$, then the Mercator series in \eqref{series_logA} is the principal logarithm of $\boldsymbol A$.
	
	\item If the real matrix $\boldsymbol A$ has no eigenvalues in $(-\infty,0]$, then for any $\alpha\in[-1,1]$ we have $\log(\boldsymbol A^\alpha) = \alpha\log(\boldsymbol A)$.
\end{itemize}


\subsection{Facts on the embeddable Markov matrices}\label{M_matrix_sto}

Let us first recall the definition of a well-known class of embeddable transition matrices. A  square matrix $\boldsymbol A$ is a nonsingular \emph{$M-$matrix} if $\boldsymbol A = s \boldsymbol I - \boldsymbol B$ with $\boldsymbol B$ a matrix with nonnegative elements and $s> \rho(\boldsymbol B)$. See, e.g., \cite{HL2011}, \cite{VB2018}. It is a standard property that the inverse of a nonsingular $M-$matrix has nonnegative entries. Moreover, the following property holds~: if the stochastic matrix $\boldsymbol P$ is the inverse of an $M-$matrix then $\boldsymbol P^{1/\ell}$ exists and is stochastic for all $\ell$.  See, e.g., \cite[Th. 3.6]{HL2011}



In the second part of this section, we recall the conditions for the uniqueness of the generator of embeddable matrices. 
Recall that a Markov chain with transition matrix $\boldsymbol P$ is \emph{embeddable} (or the transition matrix $\boldsymbol P$ is embeddable) if and only if there exists a matrix $\boldsymbol G$ with nonnegative off-diagonal entries whose row sums all vanish, such that $\boldsymbol P = \exp(\boldsymbol G)$. The matrix $\boldsymbol G$ is usually called a \emph{generator} matrix, and alternative terminologies are intensity or rate matrix, skeleton, or $Q-$matrix. An embeddable matrix does not necessarily possess a unique generator. This identifiability problem is related to the  uniqueness of the real logarithm of a matrix. Some  sufficient conditions for the uniqueness are:
\begin{enumerate}
	\item If  $\boldsymbol P$ has distinct, real, and positive eigenvalues then the only real logarithm, and hence the only candidate generator, is the principal logarithm. See \cite[Section 2.3]{H2008}. 
	\item If  $\boldsymbol P$ 	admits a generator matrix, it is unique if and only if all the eigenvalues of   $\boldsymbol P$
	are positive, real and, no Jordan block of  $\boldsymbol P$
	belonging to any eigenvalue appears more than once. See \cite[Theorem 2]{C1966}. 
	
	\item Let  $\boldsymbol P$  be an embeddable Markov matrix. Then if any of the conditions 
		$$(a) \quad \min_{i}  \boldsymbol P _{ii} > \frac{1}{2};\qquad \qquad (b) \quad 
		\left[\min_{i}  \boldsymbol P _{ii} \right] \operatorname{det} (\boldsymbol{P}) > \exp(-\pi) \prod_{i}  \boldsymbol P _{ii},
		$$
holds true, the generator of $\boldsymbol P$ is unique. 
 See \cite{C1972,C1973}.
\end{enumerate}


\subsection{Proof of the Markov property for $(Y_k)_{k\geq 1 }$}\label{sec_Markov_Y}

Below, the symbol $\lesssim$ means the left side is bounded above  by a  constant times the right side. Moreover, $a\gtrsim b$ means $b\lesssim a$, and $a\sim b$ means $a\lesssim b$ and $b\lesssim a$. $\overline C$, $\mathfrak C $, $\mathfrak c, \mathfrak c^\prime $ are constants, possibly different from line to line, but not depending on $N$,  $\z$ or $c$ (the bandwidth factor).  

\medskip

\begin{proof}[Proof of the Markov property for $(Y_k)_{k\geq 1 }$]
Indeed, for any $k\geq 1$ we have
\begin{multline}
	\mathbb P (Y_{k+1} = j \mid Y_{k} = i,  Y_{k-1}, \ldots Y_1 ,\Z=\z ) 
	\\ \hspace{-5.2cm}= \sum_{\ell \geq 1}     \mathbb P (Y_{k+1} = j, \tau_{k+1} = \ell \mid Y_{k} = i,  Y_{k-1}, \ldots Y_1, \Z=\z ) \\
	\hspace{.8cm}  =  \sum_{\ell \geq 1}    \sum_{\ell_k,\ldots,\ell_1\geq 1} \mathbb P (Y_{k+1} = j, \tau_{k+1} = \ell \mid Y_{k} = i,  Y_{k-1},\ldots Y_1, \tau_k=\ell_k,\dots,\tau_1=\ell_1, \Z=\z ) \\
	\times \mathbb P ( \tau_k=\ell_k,\dots,\tau_1=\ell_1 \mid Y_{k-1}, \ldots Y_1, \Z=\z )
	\\  \hspace{-4cm}=\sum_{\ell \geq 1} \mathbb P (Y_{k+1} = j , \tau_{k+1} = \ell \mid Y_{k} = i, \Z=\z) \\ 
	\hspace{1.5cm}\times \sum_{\ell_k,\ldots,\ell_1\geq 1} \mathbb P ( \tau_k=\ell_k,\dots,\tau_1=\ell_1 \mid Y_{k-1}, \ldots Y_1, \Z=\z )\\
	\hspace{-.9cm} =\sum_{\ell \geq 1} \mathbb P (Y_{k+1} = j , \tau_{k+1} = \ell \mid Y_{k} = i, \Z=\z)
	\\  = \mathbb P (Y_{k+1} = j \mid Y_{k} = i, \Z=\z).
\end{multline}
\end{proof}

\medskip

\begin{proof}[Proof of identity \eqref{power_Y1}]
For any $i,j\in\mathcal S$, $\ell\geq 1$ and $\z\in\mathcal Z$, we can write
\begin{multline}
	\mathbb P (Y_{k+1} = j \mid Y_{k} = i,  \tau_{k+1} = \ell, \Z=\z) 
	\\ \hspace{-1.8cm} = \sum_{t\geq 1}\mathbb P (Y_{k+1} = j \mid Y_{k} = i,  \tau_{k+1} = \ell, T_k=t, \Z=\z)\\ \hspace{-1cm}\times \mathbb P (T_k=t \mid Y_{k} = i,  \tau_{k+1} = \ell,  \Z=\z) 
	\\ \hspace{1.1cm} =  \sum_{t\geq 1}\mathbb P (X_{T_{k} + \ell} = j \mid X_{T_k} = i,  T_{k+1} = t+ \ell, T_k=t, \Z=\z) \\ \hspace{1cm}\times \mathbb P (T_k=t \mid Y_{k} = i,  \tau_{k+1} = \ell,  \Z=\z) \\
	\hspace{1.9cm}= \left(\boldsymbol P^{\ell}(\z)\right)_{ij}\times \sum_{t\geq 1}\mathbb P (T_k=t \mid Y_{k} = i,  \tau_{k+1} = \ell,  \Z=\z)
	\\= \left(\boldsymbol P^{\ell}(\z)\right)_{ij}.
\end{multline}
This proves the identity \eqref{power_Y1}.
\end{proof}


\subsection{Proofs of the results}\label{add_prooofs}

\begin{proof}[Proof of Lemma \ref{lemma_ind}]
	If $(\tau_k)_{k\geq 1}\perp X \mid \Z$, and the $\tau_k$, $k\geq 1$, are conditionally i.i.d. given $\Z=\z$, for any $\z\in\mathcal Z$, then we have $\mathbb P (\tau_{k+1} = \ell \mid Y_{k} = i, \Z=\z)=\mathbb P (\tau_{k+1} = \ell \mid \Z=\z)$, for any $ i\in \mathcal S$ and any  $\ell \geq 1$, $k\geq 1$. Moreover, $\mathbb P (\tau_{k+1} = \ell \mid \Z=\z)$ does not depend on $k$. Thus Assumption \ref{assump4} holds true.  By our definitions and the conditional independence assumption $(\tau_k)_{k\geq 1}\perp X \mid \Z$, for any $ i_{k-1},\ldots,i_0\in\mathcal S$ and $\ell\geq 1$,
	the conditional distribution of 
	$$
	Y_{k+1} \mid \tau_{k+1}=\ell, Y_k=i,\tau_k=\ell_k,Y_{k-1}=i_{k-1},\ldots,\tau_1=\ell_1,Y_0=i_0,\Z=\z,
	$$ 
	coincides with that of $X_{t+\ell} \mid \tau_{k+1}=\ell, X_t=i,\tau_k+\cdots+\tau_1+T_0=t,\Z=\z$, for any $\ell_k,\ldots,\ell_1,t$ such that  $t=\ell_k+\cdots+\ell_1+T_0$, and is determined by $i$ and  the $j$-th row of the matrix $\boldsymbol{P}^{\ell}(\z)$. Gathering facts, we deduce that the Assumptions \ref{assump2} and \ref{assump3} hold true.  \end{proof}
	
	\medskip

	\begin{proof}[Proof of Proposition \ref{stoch1}]
		Let $\Omega_N=\sum_{m=1}^N m^{\beta}$ with $\beta>0$ from Assumption \ref{assump1_th}-\ref{mker1}, and recall that  $w_{N,m}= m^{\beta} = \omega_{N,m}\Omega_N^{-1}$. Since $\beta<1$,  by the technical Lemma \ref{serie_Riemann}, we have $\Omega_N\sim N^{1+\beta}$. Let $i,j\in\mathcal S$, $1\leq \ell \leq L$ be arbitrarily fixed.	For any $\z\in \mathcal Z$,    let
		\begin{equation}\label{generic_app_b}
			G_m(\z;c)=  
			    a_{m} \mathcal I_{m}  h_m^{p}\boldsymbol K_{h_m}  \left( \Z_m-\z  \right), \quad \text{ where } \quad a_{m} = w_{N,m} h_m^{-p}, \quad \mathcal I_{m}=\sum_{k=1}^{M_m}\mathcal I_{m,k},
		\end{equation}	 
		with $	\mathcal I_{m,k}=\mathcal I_{m,k}(i,j;\ell) =  \mathds{1} \left\{ Y_{m,k} = j, \tau_{m,k}=\ell, Y_ {m,k-1} = i \right\},
		$
		and $c$ is the constant from the bandwidths $  h_m$. By our assumptions, $M_m$, $m\geq 1$  is a sequence of independent, positive integer-valued random variable bounded by $L$. Therefore, $\mathcal I_m$, $m\geq 1$  is also a sequence of independent, positive integer-valued random variable bounded by $L$.
		In view of the  definition of $G_m$, we write $ \widehat{U}_{N}(\z;c)  $ instead of $\widehat{U}_{N}(\z)$, and we have $ \widehat{U}_{N}(\z;c) = \Omega_N^{-1}  \sum_{m=1}^N G_m(\z;c)$.

		By our assumptions, $a_m\sim m^{\alpha p +\beta }$.   
		We apply Bernstein's inequality  which states that, for any   $X_1,\ldots, X_N$  zero mean independent variables, such that $|X_m|\leq C_N$, $\forall 1\leq m\leq N$, we have
		$$
		\forall v\geq 0, \quad \PP \left(  \left|\sum_{m=1}^N X_m\right| \geq v\right)\leq 2 \exp\left( - \frac{v^2/2}{\sigma^2_N+C_Nv/3}\right),\quad \text{ where } \quad \sigma_N^2 = \sum_{m=1}^N \EE (X_m^2).
		$$
		See, for example, \citep[Theorem 2.8.4]{V2018}. Let 
		\begin{equation}\label{xi_def}
			X_m = X_m(\z;c) = G_m(\z;c) - \EE[G_m(\z;c)],
		\end{equation}
		which implies  $C_N \sim \max\{a_1,\ldots ,a_N\} \sim N^{\alpha p +\beta }$. By Lemma \ref{serie_Riemann} and  since  $\mathcal I_m$  is bounded by $L$, we have 
		\begin{equation}\label{sigma_rate}
			\sigma^2_N \sim \sum_{m=1}^N m^{\alpha p +2\beta }\;\EE [\mathcal I^2_m h_m^p\{\boldsymbol K_{h_m}  \left( \Z_m-\z  \right)\}^2  ] 
			\lesssim L^2 N^{\alpha p +2\beta+1 }.
		\end{equation}
		The factor $L^2$  is a rough bound for the right hand side of this inequality. The bound  can be refined to $L \sup_{\z\in \mathcal Z} \PP(Y_{m,k} = j, \tau_{m,k}=\ell, Y_ {m,k-1} = i\mid \Z=\z)$, but this would not change the rate of convergence.
		Taking $v= \overline C (\log N)^{1/2} N^{(\alpha p +1)/2 + \beta}$ with $\overline C$ some large constant, since $(\alpha p +1)/2< 1$, we have
		$$
		\frac{v^2/2}{\sigma^2_N+C_Nv/3} \gtrsim \frac{\overline C^2 N^{(\alpha p +1) + 2 \beta}\times \log N}{N^{\alpha p +2\beta+1 } + N^{\alpha p +\beta } \times \overline C  (\log N)^{1/2}N^{(\alpha p +1)/2 + \beta} }\sim \overline C \log N.
		$$
		Noting that $v\Omega_N^{-1} \sim  (\log N)^{1/2} N^{(\alpha p -1)/2 }$, we deduce that, for any $\z\in\mathcal Z$,
		\begin{multline}\label{az56}
			\PP \left(\left|\widehat{U}_{N}(\z;c)- \EE\big [\widehat{U}_{N}(\z;c)\big] \right| \geq\overline C (\log N)^{1/2} N^{(\alpha p -1)/2 } \right) 
			\\=	\PP \left( \Omega_N^{-1}  \left|\sum_{m=1}^N X_m(\z;c)\right| \geq \overline C (\log N)^{1/2} N^{(\alpha p -1)/2 } \right) 
			\leq 2 \exp(-\mathfrak c \times \overline C \log N),
		\end{multline}
		with $\mathfrak c$ a constant depending only on $\alpha,\beta$, the  $\underline c,\overline c$ giving the bandwidth range in condition \ref{mker1} of Assumption \ref{assump1_th}, the maximal length $L$ of a sample path,  and  the uniform norm of $\boldsymbol K$.   The pointwise rate from Proposition \ref{stoch1} follows.

		To obtain the uniformity with respect to  $\z$ and the constant $c$ in the bandwidth, we consider a  grid of $\mathcal G$ points $(\z_{\ell^\prime }, c_{\ell^{\prime\prime}})$ in the set $ \mathcal Z\times [\underline c , \overline c]$, with $\mathcal G\sim N^{\gamma}$, $\gamma \geq 2$. The grid is constructed as the Cartesian product between an equidistant grid in the hyperrectangle $\mathcal Z_c\times [\underline c, \overline c]$ and the discrete finite set $\mathcal Z_d$. Using the Lipschitz conditions on the kernel $ K(\cdot)$, for any $\z,c$, indices $ l^\prime,  l^{\prime\prime}$ exist such that 
		$$
		|G _m(\z;c) - G _m(\z_{ l^\prime }; c_{ l^{\prime\prime}})|\leq \mathfrak C N^{-1}, \qquad  \forall 1\leq m \leq N,
		$$ 
		where $\mathfrak C$ is some constant. As a consequence, 
		$$
		\left|\{G _m(\z;c) - \EE[G _m(\z;c) ]\} - \{G _m(\z_{ l^\prime }; c_{ l^{\prime\prime}})- \EE[G _m(\z_{ l^\prime }; c_{ l^{\prime\prime}})]\}\right|\leq \mathfrak C N^{-1}, \quad   \forall 1\leq m \leq N.
		$$ 
		From this, inequality \eqref{az56}   and Boole's union bound inequality, we deduce that constants $\mathfrak c, \mathfrak c^\prime, \mathfrak C$ exist such that, 
		\begin{multline}
			\PP \Big(\sup_{\z,c}\left|\widehat{U}_{N}(\z;c)- \EE\big[\widehat{U}_{N}(\z;c)\big] \right| \geq \overline  C (\log N)^{1/2} N^{(\alpha p -1)/2 } \Big) 
			\\	\leq 2  |\mathcal G|   \exp(-\mathfrak c \times \overline C \log N)\\ \leq 2 \exp\{\mathfrak c^\prime  \gamma  \log N-\mathfrak c  \overline  C \log N\}\rightarrow 0.
		\end{multline}
		Here, $|\mathcal G|$ denotes the cardinal of $\mathcal G$. The  convergence to zero is guaranteed by the fact that  $\overline C$ can be large. This proves the result. \end{proof} 
	
	\medskip

\begin{proof}[Proof of Corollary \ref{stoch1_cor}]
	Let $a_N:=    N^{-(1-\alpha p)/2}\sqrt{\log N } $. 
We first show that $\EE [\widehat U_{T,N}(i;\z,\ell)] $ is uniformly  bounded away 	from zero. By the definitions and Assumption \ref{assump1_th}-\ref{mxia2c}, a constant $C>0$ exists such that 
\begin{multline}
\EE \left[\widehat U_{T,N}(i;\z,\ell)\right] = \Omega_N^{-1}  \sum_{m=1}^N \sum_{j\in\mathcal S} \EE \left[ G_m(\z;c) \right] \\\hspace{-2cm}  = \Omega_N^{-1}  \sum_{m=1}^N w_{N,m} \EE\left[\EE\left\{\sum_{j\in\mathcal S}  \mathcal I_m \mid \Z_m\right\} \boldsymbol K_{h_m}  \left( \Z_m-\z  \right)\right]\\
\hspace{1cm}  =\Omega_N^{-1}  \sum_{m=1}^N w_{N,m} \int \PP (  \tau_{k}=\ell, Y_ {k-1} = i \mid \Z=\z^\prime) f_{\Z}(\z^\prime )\boldsymbol K_{h_m}  \left( \z^\prime-\z  \right) d\mu(\z^\prime) \\ >C\Omega_N^{-1}  \sum_{m=1}^N w_{N,m} \int \boldsymbol K_{h_m}  \left( \z^\prime-\z  \right) d\mu(\z^\prime) =C, 
\end{multline}
$\forall i\in\mathcal S$, $\ell \in[1,L]$, $\z\in\mathcal Z$, where $\mu$ is the product measure between the Lebesgue measure in $\mathbb R^p$ and the counting measure, and $\mathcal I_m$ is defined in \eqref{generic_app_b}.  Summing over the values of $j$, we deduce that $\EE [\widehat U_{B,N}(i;\z,\ell)] $ is also uniformly  bounded away 	from zero. From this and Proposition \ref{stoch1}, we can write
\begin{multline}
	\widehat U_{B,N}(i;\z,\ell)	^{-1} -\EE \left[\widehat U_{B,N}(i;\z,\ell)\right] ^{-1} \\ \hspace{-2cm}=\left\{ \EE \left[\widehat U_{B,N}(i;\z,\ell)\right] +  O_\PP(a_N)\right\}^{-1}-\EE \left[\widehat U_{B,N}(i;\z,\ell)\right] ^{-1}\\	
	=  O_\PP(a_N) \left\{ \EE \left[\widehat U_{B,N}(i;\z,\ell)\right] +  O_\PP(a_N)\right\}^{-1}\EE \left[\widehat U_{B,N}(i;\z,\ell)\right] ^{-1}	=  O_\PP(a_N),
\end{multline}
uniformly with respect to $i$, $\ell$ and $\z$. By Proposition \ref{stoch1} we also have
$$
\widehat U_{T,N}(i,j;\z,\ell)  -\EE \left[\widehat U_{T,N}(i,j;\z,\ell)\right]  =  O_\PP(a_N),
$$
uniformly with respect to $i,j$, $\ell$ and $\z$. Gathering facts, the justification of the first part of the Corollary is complete. The second part, concerning $\widehat \pi_\ell (\z)$, follows by Proposition \ref{stoch1_cor} and the identity \eqref{nice_pi}. \end{proof}
 
 \medskip

 \begin{proof}[Proof of Proposition \ref{bias1}]
 	Recall that 
 	$G_m(\z;c)$ is defined in \eqref{generic_app_b}. 
 	For any $\forall \z\in\mathcal Z_\epsilon$, by a change of variables and the second order Taylor expansion,   we get 
 	\begin{multline}\label{T_E}
 		\EE [w^{-1}_{N,m}G_m (\z;c)]:= \EE\left[\EE\left\{  \mathcal I_m \mid \Z_m\right\} \boldsymbol K_{h_m}  \left( \Z_m-\z  \right)\right] \\
 		=\int \PP (Y_ {k} = j,  \tau_{k}=\ell, Y_ {k-1} = i \mid \Z=\z^\prime) f_{\Z}(\z^\prime )\boldsymbol K_{h_m}  \left( \z^\prime-\z  \right) d\mu(\z^\prime)
 		\\= \PP (Y_ {k} = j,  \tau_{k}=\ell, Y_ {k-1} = i \mid \Z=\z) f_{\Z}(\z ) \\
 		+ \frac{  h_m^2}{2} \mu_2 (\boldsymbol K) \operatorname{Trace}(\mathcal H_G(\z)) \{1+o(1)\},
 	\end{multline}
 	provided $m$ is sufficiently large. Here, $\mathcal H_G(\z)$ denotes the Hessian matrix of the map  $\z\mapsto \PP (Y_ {k} = j,  \tau_{k}=\ell, Y_ {k-1} = i \mid \Z=\z) f_{\Z}(\z )$,  which, by our Assumptions \ref{assump1_th}-(\ref{mxia2}, and \ref{mxia2b}) is uniformly bounded. 
 	The calculations in \eqref{T_E} are valid as soon as $m$ is sufficiently large such that $\z+ h_m\mathbf u\in\mathcal Z_c$. 
 	Next, since 
 	$$
 	\EE \left[\widehat{U}_{N}(\x)\right]= \Omega_N^{-1}\sum_{m=1}^N w_{N,m}  \EE [w^{-1}_{N,m} G_m (\z;c)],  
 	$$
 	and, by Lemma \ref{serie_Riemann}, $\Omega_N \sim N^{\beta + 1}$ and  $\sum_{m=1}^N \omega_{N,n}  h_m^2\sim N^{\beta + 1 -2\alpha}$, the result follows. 
 	\end{proof}

 \medskip
 
\begin{proof}[Proof of Proposition \ref{final_touch}]
The entries of $\widehat{\boldsymbol A}_\ell (\z)$ converge uniformly to those of ${\boldsymbol A}_\ell (\z)$, and this implies that the two spectra gets close, uniformly with respect to $\z\in\mathcal Z_\epsilon$. More precisely,  
for any two $S\times S-$matrices $\boldsymbol A$ and $\boldsymbol B$, with respective spectra $\{\lambda_{\boldsymbol A, 1},\ldots,\lambda_{\boldsymbol A, S}\}$ and $\{\lambda_{\boldsymbol B,1},\ldots,\lambda _{\boldsymbol B, S}\}$, it holds that 
$$
\max\{s_{\boldsymbol A}(\boldsymbol B), s_{\boldsymbol B}(\boldsymbol A)\}\leq \left\{\|\boldsymbol A\| + \|\boldsymbol B \|\right\}^{1- 1/S} \|\boldsymbol A - \boldsymbol B\|^{1/S},
$$
where $\|\cdot\|$ is the spectral norm and  $s_{\boldsymbol A}(\boldsymbol B)$ is the spectral variation, that is 
$$
s_{\boldsymbol A}(\boldsymbol B) = \max_{j} \min_{i} \left| \lambda_{\boldsymbol A, i}   - \lambda_{\boldsymbol B, j} \right|.
$$ 
See \cite{E1985}. We  apply this result with $\widehat{\boldsymbol A}_\ell (\z)$  and ${\boldsymbol A}_\ell (\z)$.
Then, using Assumption \ref{assump5}, the fact that the eigenvalues of ${\boldsymbol A}_\ell (\z)$ are the powers of the eigenvalues of ${\boldsymbol P} (\z)$ (e.g., \citep[Theorem 1.13-(d)]{H2008}), and the uniform convergence in Corollary \ref{cor_lambda1}, we deduce that there exists a (small) constant $\mathfrak c>0$ such that the probability of the event $\{\min_{1\leq \ell \leq L}\min_i\inf_{\z\in\mathcal Z_\epsilon} |\lambda_{\widehat{\boldsymbol A}_\ell (\z),i}(\z)|\geq \mathfrak c\}$ tends to 1. Moreover, it also holds 
\begin{equation}\label{cst_event}
	\PP	\left(\min_{1\leq \ell \leq L}\inf_{\z\in\mathcal Z_\epsilon}d_{\rm H}(\sigma[\widehat{\boldsymbol{A}}_{\ell}(\z)],\mathbb R^{-})\geq \mathfrak c\right)\rightarrow 1.
\end{equation}
This guarantees that the real logarithm of $\widehat{\boldsymbol{A}}_{\ell}(\z)$ is well defined for all $\z\in \mathcal Z_\epsilon$ with probability tending to 1.

Next, let us recall the expression of the Frechet derivative of the logarithm function that can be deduced from the integral representation \eqref{topm}. With the notation from \citep[pp 272]{H2008}, the Frechet derivative at $\boldsymbol A$ applied to a matrix $\boldsymbol \Delta$ is 
$$
L(\boldsymbol A, \boldsymbol \Delta) = \int_0^1 [t(\boldsymbol A-\boldsymbol I ) +  \boldsymbol I]^{-1} \boldsymbol \Delta [t(\boldsymbol A-\boldsymbol I ) +  \boldsymbol I]^{-1} dt.
$$
Taking $\boldsymbol \Delta= \widehat{\boldsymbol{A}}_{\ell}(\z)-{\boldsymbol{A}}_{\ell}(\z)$, we deduce that for any matrix norm, a constant $\underline C>0$ exists such that 
\begin{equation}\label{lip_log}
\left\| \log(\widehat{\boldsymbol{A}}_{\ell}(\z))- \log({\boldsymbol{A}}_{\ell}(\z)) \right\| \leq \underline C 
\left\| \widehat{\boldsymbol{A}}_{\ell}(\z)- {\boldsymbol{A}}_{\ell}(\z) \right\| .
\end{equation}
The constant $ \underline C $ depends on $\epsilon$ defining $\mathcal Z_\epsilon$, the constant $ \mathfrak c $ in \eqref{cst_event} (and on the  matrix norm considered), but not on $\z\in\mathcal Z_\epsilon$ or $\ell\in[1,L]$. From inequality \eqref{lip_log} we deduce that $\log(\widehat{\boldsymbol{A}}_{\ell}(\z))$ inherits the rate of convergence of $\widehat{\boldsymbol{A}}_{\ell}(\z)$. Finally, the regularized version of $\log(\widehat{\boldsymbol{A}}_{\ell}(\z))$, it suffices to note that the entries of $\widehat{\boldsymbol{B}}_{\ell}(\z)$ are obtained by simple transformations of the entries of $\log(\widehat{\boldsymbol{A}}_{\ell}(\z))$. Then it can be shown that, for any matrix norm
\begin{multline}\label{lip_log2}
	\left\| \log(\widehat{\boldsymbol{B}}_{\ell}(\z))- \log({\boldsymbol{A}}_{\ell}(\z)) \right\|\leq \left\| \log(\widehat{\boldsymbol{B}}_{\ell}(\z))- \log(\widehat{\boldsymbol{A}}_{\ell}(\z)) \right\|  \\ +\left\| \log(\widehat{\boldsymbol{A}}_{\ell}(\z))- \log({\boldsymbol{A}}_{\ell}(\z)) \right\| 
	=
O_{\mathbb P}	\left(\left\| \widehat{\boldsymbol{A}}_{\ell}(\z)- {\boldsymbol{A}}_{\ell}(\z) \right\|\right) .
\end{multline}
Then the statement of Proposition \ref{final_touch} follows. 
\end{proof}
 
 \medskip

  \begin{proof}[Proof of Corollary \ref{final_touch}]
This is a direct consequence of Lipschitz property of the matrix exponential function, see \cite{H2008}, and the uniform convergence rate of $ \log(\widehat{\boldsymbol{B}}_{\ell}(\z))$ derived in Proposition \ref{final_touch}. 
 \end{proof}
 
 \medskip

\begin{lemma}\label{serie_Riemann}
	Let  $\rho \geq 0$, $\varrho >0$, $\varrho \neq 1$, 
	and $ N>1$. Then
	\begin{equation}\label{sumR1}
		\frac{( N+1)^{1-\varrho}-1}{1-\varrho} 
		\leq \sum_{ m=1}^ N  m^{-\varrho} 
		\leq \frac{ N^{1-\varrho}-\varrho}{1-\varrho} 
		\quad \text{	and  } \quad 
		\frac{ N^{1+\rho}+\rho}{1+\rho} \leq \sum_{ m=1}^ N  m^{\rho } \leq \frac{( N+1)^{1+\rho}-1}{1+\rho}.
	\end{equation}
	
\end{lemma}

\medskip

\begin{proof}[Proof of Lemma \ref{serie_Riemann}]
	For any $ m \geq 1$ and $x\in [ m, m+1]$, we have
	$
	( m+1)^{-\varrho} \leq x^{-\varrho} \leq  m^{-\varrho}.
	$
	Then 
	$$
	\sum_{ m=2}^{ N+1}  m^{-\varrho} = \sum_{ m=1}^{ N} ( m+1)^{-\varrho} \leq \int_{1}^{ N+1} x^{-\varrho} dx = \frac{1}{1-\varrho}\left[( N+1)^{1-\varrho} - 1
	\right]\leq \sum_{ m=1}^{ N}  m^{-\varrho},
	$$
	from which the first part of \eqref{sumR1} follows. The case $\rho=0$ is obvious.
	Finally, since $\rho >0$, for any $ m \geq 1$ and $x\in [ m, m+1]$, we have
	$
	 m^{\rho} \leq x^{\rho} \leq ( m+1)^{\rho}.
	$
	Thus
	$$
	\sum_{ m=1}^{ N}  m^{\rho} \leq \int_{1}^{ N+1} x^{\rho} dx = \frac{1}{1+\rho}\left[( N+1)^{1+\rho} - 1
	\right]\leq \sum_{ m=1}^{ N} ( m+1)^{\rho} = \sum_{ m=2}^{ N+1}  m^{\rho} ,
	$$
	and  the second part of \eqref{sumR1} follows.  \end{proof}

\bibliographystyle{abbrv}
\bibliography{references-chapitre2}

\end{document}